\documentclass[letterpaper]{article}

\usepackage{supp}

\usepackage[utf8]{inputenc} 
\usepackage[T1]{fontenc}    
\usepackage{hyperref}       
\usepackage{url}            
\usepackage{booktabs}       
\usepackage{amsfonts}       
\usepackage{nicefrac}       
\usepackage{microtype}      
\usepackage{xcolor}         

\hypersetup{colorlinks=true, citecolor=blue, linkcolor=red}

\title{Differentially Private Condorcet Voting}

\author{}

\usepackage{amsmath,amsfonts,amssymb,amsthm}
\usepackage{booktabs}
\usepackage{pifont}
\usepackage[ruled,linesnumbered,vlined]{algorithm2e}
\usepackage{float}
\usepackage{graphicx}
\usepackage{mathrsfs}
\usepackage{multirow,multicol}
\usepackage{subfigure}
\usepackage{enumitem}
\usepackage[all]{xy}

\usepackage{fancyhdr}
\newtheorem{theorem}{Theorem}
\newtheorem{definition}{Definition}
\newtheorem{proposition}{Proposition}
\newtheorem{lemma}{Lemma}

\theoremstyle{definition}
\newtheorem*{remark}{Remark}

\begin{document}


\def\p{\mathbb{P}}
\def\lap{\operatorname{Lap}}
\def\sgn{\operatorname{Sgn}}
\def\LA{\mathcal{L}(A)}
\def\erm{ {\rm e} }
\def\drm{ {\rm d} }
\def\CW{\operatorname{CW}}
\def\PDF{f_\lambda}
\def\CDF{F_\lambda}
\def\CMLAP{\operatorname{CM}^\text{\rm LAP}}
\def\CMEXP{\operatorname{CM}^\text{\rm EXP}}
\def\CMRR{\operatorname{CM}^\text{\rm RR}}
\def\CMRand{\operatorname{CM}^{\texttt{Rand}}}
\def\RandSet{ \{ \text{\rm LAP}, \text{\rm EXP}, \text{\rm RR} \} }

\renewcommand\ge{\geqslant}

\maketitle

\begin{abstract}
	Designing private voting rules is an important and pressing problem for trustworthy democracy. In this paper, under the framework of differential privacy, we propose a novel famliy of randomized voting rules based on the well-known Condorcet method, and focus on three classes of voting rules in this family: Laplacian Condorcet method ($\CMLAP_\lambda$), exponential Condorcet method ($\CMEXP_\lambda$), and randomized response Condorcet method ($\CMRR_\lambda$), where $\lambda$ represents the level of noise. We prove that all of our rules satisfy absolute monotonicity, lexi-participation, probabilistic Pareto efficiency, approximate probabilistic Condorcet criterion, and approximate SD-strategyproofness. In addition, $\CMRR_\lambda$ satisfies (non-approximate) probabilistic Condorcet criterion, while $\CMLAP_\lambda$ and $\CMEXP_\lambda$ satisfy strong lexi-participation. Finally, we regard differential privacy as a voting axiom, and discuss its relations to other axioms.
\end{abstract}

\section{Introduction}
\label{sec: intro}

\noindent Voting is a commonly used method for group decision making, where voters submit their preferences over a set of alternatives, and then a voting rule is applied to choose the winner. A major and classical paradigm behind the design and analysis of voting rules is the {\em axiomatic approach}~\cite{Plott76:Axiomatic}, under which voting rules are evaluated by their satisfaction to various normative properties, known as {\em (voting) axioms}. For example, the {\em Condorcet criterion} requires that whenever there exists a {\em Condorcet winner}, which is the alternative that beats all other alternatives in their head-to-head competitions, it must be selected as the winner.

Recently, privacy in voting has become a critical public concern. There are a series of works on examining the {\em differential privacy (DP)}~\cite{Dwork06D} of voting~\cite{shang2014application,hay2017differentially,yan2020private}.
These works mainly focused on applying several randomized mechanisms to existing voting rules, proving upper bounds on the privacy-preserving level (also called {\em privacy budget}, denoted by $\epsilon$ throughout the paper), and then evaluating the utility loss (measured by accuracy or mean square error) due to randomness.
However, the upper bounds on privacy in most of them are not tight, which means that the exact privacy-preserving level of the mechanisms is unclear.
Moreover, we are not aware of a previous work on making voting private while maintaining the satisfaction to desirable voting axioms beyond strategyproofness~\cite{DBLP:conf/ijcai/Lee15}.
Therefore, the following question remains largely open.
\begin{center}
	{\em How can we design private voting rules that satisfy desirable axiomatic properties?}
\end{center}


\paragraph{Our contributions.} We propose a novel class of randomized voting rules, denoted by $\CMRand_\lambda$, based on the celebrated {\em Condorcet method}, which chooses the Condorcet winner when it exists, where \texttt{Rand} is a randomized function (called a {\em mechanism} in DP literature) that introduces noises to pairwise comparisons between alternatives, and $\lambda$ represents the level of noise. To choose a winner, $\CMRand_\lambda$ applies \texttt{Rand} with parameter $\lambda$ to the pairwise comparisons for the input profile until a Condorcet winner appears, and then chooses it as the winner.

We focus on three classes voting rules in this family, namely $\CMLAP_\lambda$, $\CMEXP_\lambda$, and $\CMRR_\lambda$, which are obtained by applying the Laplace mechanism, exponential mechanism, and randomized response mechanism, respectively. Under these mechanisms, while it may take exponentially many iterations to obtain the winner by definition, we show that the winner can be efficiently sampled (Lemma~\ref{lem: mech2scf}).

\begin{table*}[ht]
	\centering
	\begin{tabular}{cccccccc}
		\toprule
		                 & p-Condorcet & $\alpha$-p-Condorcet                                                      & p-Pareto  & a-Mono.   & $\alpha$-SD-SP         & Lexi-Par. & Strong Lexi-Par. \\
		\midrule
		$\CMRR_\lambda$  & \ding{51}   & $\erm^{\lambda}$                                                          & \ding{51} & \ding{51} & $\erm^{(2-2m)\lambda}$ & \ding{51} & \ding{55}        \\
		$\CMEXP_\lambda$ & \ding{55}   & $\frac{ 1+\erm^{\lambda/2} }{ \left( 1+\erm^{-\lambda/2} \right)^{m-1} }$ & \ding{51} & \ding{51} & $\erm^{(2-2m)\lambda}$ & \ding{51} & \ding{51}        \\
		$\CMLAP_\lambda$ & \ding{55}   & $2\erm^\lambda\left(1-\frac{\erm^{-\lambda}}{2}\right)^{m-1}$             & \ding{51} & \ding{51} & $\erm^{(2-2m)\lambda}$ & \ding{51} & \ding{51}        \\
		\bottomrule
	\end{tabular}
	\caption{The satisfaction of our mechanisms to the voting axioms, where ``\ding{51}'' indicates that the row rule satisfies the column axiom, and ``\ding{55}'' indicates that the rule does not satisfy the axiom. The expressions in the table represent the level of satisfaction to the approximate axioms (the $\alpha$ in $\alpha$-p-Condorcet and $\alpha$-SD-SP).}
	\label{tab: mechanism-criteria}
\end{table*}

Our main technical contributions are three-fold. First, we prove that all the three classes of voting rules are differentially private by characterizing the upper and lower bounds on the privacy budget $\epsilon$ (Theorem \ref{thm: dp}). Second, we study the satisfaction of our voting rules to probabilistic variants to Condorcet criterion (p-Condorcet, requiring the winning rate of the Condorcet winner is not lower than the other alternatives), Pareto efficiency (p-Pareto, which requires the winning rate of $a$ is not lower than $b$, if $a$ Pareto dominates $b$), monotonicity (a-monotonicity, which ensures the winning rate of each alternative does not decrease when her ranking is lifted by any voter simply), strategyproofness (SD-strategyproofness, SD-SP for short, which ensures that no voter can benefit herself in the sense of stochastic dominance by changing her vote), and participation (lexi-participation, which ensures that no voter can improve the result of the voting lexicographically by withdrawing her vote). Besides, we consider the approximate version of p-Condorcet ($\alpha$-p-Condorcet, Definition \ref{def: alpha-pCond}) and SD-SP ($\alpha$-SD-SP, Definition \ref{def: alpha-SDSP}), and the strong version of lexi-participation (Definition \ref{def: strong-lexi}). Our results suggest that $\CMLAP_\lambda$ outperforms $\CMEXP_\lambda$ in all aspects examined in the paper, while $\CMRR_\lambda$ sometimes achieves better p-Condorcet but only satisfies standard lexi-participation, instead of the strong version (Theorems \ref{thm: sat-pCond} - \ref{thm: EXP-LAP-lexi-participation}). The results in the second part are summarized in Table \ref{tab: mechanism-criteria}. Third, we investigate the relations between DP and the voting axioms. We prove that Condorcet criterion and Pareto efficiency are incompatible with DP, and capture the upper bounds of satisfaction to p-Condorcet under $\epsilon$-DP (Proposition \ref{prop: Condorcet-incomp-DP} - \ref{prop: DP-Cond-upperbound}). Besides, we show that DP guarantees a lower bound of satisfaction to SD-strategyproofness (Proposition \ref{prop: dp-strategyproof}).

\paragraph{Related work and discussions.}

To the best of our knowledge, DP was first applied to the rank aggregation problem in~\cite{shang2014application}. They analyzed the error rates and derived upper bounds on them.
Lee proposed an algorithm which is both differentially private and robust to strategic manipulation for tournament voting rules~\cite{DBLP:conf/ijcai/Lee15}.
Hay et al. used Laplace mechanism and exponential mechanism to improve the privacy of Quicksort and Kemeny-Young method~\cite{hay2017differentially}.
Kohli and Laskowski explored DP, strategyproofness, and anonymity for voting on single-peaked preferences~\cite{kohli2018epsilon}.
Torra analyzed the privacy-preserving level of random dictatorship with DP, which is a well-known randomized voting rule~\cite{torra2019random}. He investigated the condition where random dictatorship is differentially private, and improved the mechanism to achieve DP for general cases.
Yan et al. made tradeoff between accuracy and privacy in rank aggregation to achieve local DP via Laplace mechanism and randomized response~\cite{yan2020private}.

Most of the above works did not consider the tradeoffs between privacy and those desirable properties, and the privacy bounds of them are usually not tight. Ao et al. proposed the exact version of distributional DP \cite{bassily2013coupled} and studied the privacy-preserving level of several voting rules, but they did not investigate how to improve the privacy~\cite{ao2020private}.
Beyond social choice, DP has also been considered in other topics of economics, such as mechanism design \cite{pai2013privacy,xiao2013privacy}, and matching and resource allocation \cite{hsu2016private,kannan2018private}.

There is a large literature on the analysis of randomized voting \cite{Brandt2017:Rolling}, most of them studied the satisfaction to axiomatic properties, e.g., complexity of manipulation \cite{walsh2012lot}, strategyproofness \cite{aziz2014incompatibility,aziz2015universal}, Pareto efficiency \cite{brandl2015incentives,gross2017vote}, participation \cite{brandl2019welfare} and monotonicity \cite{DBLP:conf/ijcai/Brandl0S18}. The fairness properties of sortition have also been investigated \cite{benade2019no,flanigan2020neutralizing,flanigan2021fair}.

The approximation of those properties was also studied.
Procaccia discussed how much a strategyproof randomized rule could approximate a deterministic rule~\cite{procaccia2010can}.
Birrell and Pass explored the approximate strategyproofness for randomized voting rules~\cite{birrell2011approximately}.
They bounded the difference of the expectations of the utility function with a parameter, but the ratio seems to be more natural for DP.

\section{Preliminaries}
\label{sec: pre}

Let $A=\{a_1,a_2,\ldots,a_m\}$ denote a set of $m\ge 2$ alternatives. For any $n\in\mathbb N$, let $N=\{1,2,\ldots,n\}$ be a set of voters. For each $j\in N$, the vote of voter $j$ is a linear order $\succ_j\in\LA$, where $\LA$ denotes the set of all linear orders over $A$, i.e., all transitive, antireflexive, antisymmetric, and complete binary relations.
Let $P=\{\succ_1,\succ_2,\ldots,\succ_n\}$ denote the \emph{(preference) profile}.
For each $j\in N$, let $P_{-j}$ denote the profile obtained from $P$ by removing $\succ_j$.
A (randomized) voting rule is a mapping $r\colon \LA^n\to \Delta(A)$, where $\Delta(A)$ denotes the set of all probability distributions on $A$.

Given a profile $P\in \LA^n$, let $S_P[a,b]$ denote the number of voters who prefer $a$ to $b$, i.e., $S_P[a,b]=|\{j\in N: a\succ_j b\}|$. Let $ w_P[a,b]=S_P[a,b]-S_P[b,a]$ be the {\em majority margin} of $a$ over $b$. Then the {\em weighted majority graph (WMG)} of $P$ can be defined: the vertices of WMG are alternatives in $A$ and there is a directed edge from $a$ to $b$ with weight $ w_P[a,b]$ if and only if $ w_P[a,b]>0$. Similarly, letting $U_P[a,b]=\sgn( w_P[a,b])$, the {\em unweighted majority graph (UMG)} of $P$ can also be defined: the set of vertices is $A$ and there is an unweighted directed edge from $a$ to $b$ if and only if $U_P[a,b]=1$, where $\sgn$ denotes the sign function, i.e., $\sgn(x)=x/|x|$ for all $x\neq 0$ and $\sgn(0)=0$. The {\em Condorcet winner} of $P$ is an alternative $a\in A$, such that $U_P[a,b]=1$ for all $b\in A\backslash \{a\}$, denoted by $\CW(P)$. Notice that the Condorcet winner is completely determined by the UMG, we also use $\CW(U_P)$ to denote the Condorcet winner claimed by the UMG.

\paragraph{\bf Axioms of voting.} A voting rule $r$ satisfies {\em Condorcet criterion}, if $\p[r(P)=\CW (P)]=1$ holds for all profile $P$ that $\CW(P)$ exists.
The rule $r$ satisfies {\em Pareto efficiency}, if $\p[ r(P)= b]=0$ for all profile $P$, where exists $a, b\in A$ that $a\succ_j b$ for all $j\in N$.
And $r$ satisfies {\em absolute monotonicity}~\cite{DBLP:conf/ijcai/Brandl0S18}, if $\p[r(P)=a]\leqslant \p[r(P')=a]$ holds for all $P,P'$, such that $P_{-j}=P'_{-j}$, $\succ_j\neq \succ'_j$, and $\succ'_j$ is a pushup of $a$ in $\succ_j$, i.e., $\succ'_j$ raises the position of $a$ in $\succ_j$, and keeps the relative position of other alternatives unchanged.
A randomized rule $r$ satisfies {\em SD-Strategyproofness} \cite{aziz2013tradeoff}, if for all $P,P'$ and $j\in N$ that $P_{-j}=P'_{-j}$ and $\succ_j\neq \succ'_j$, $\sum_{b\succ_j a} \p[r(P)=b]\geqslant \sum_{b\succ_j a} \p[r(P')=b]$, for all $a\in A$ \footnote{In fact, absolute monotonicity and SD-strategyproof are equivalent to the nonperverseness and the strategyproofness in \cite{gibbard1977manipulation}, respectively.}.
A voting rule satisfies {\em lexi-participation} if for all $P,P'$ that $P'=P\backslash \{\succ_j\}$, there does not exist $a\in A$, such that $\p[r(P)=a]< \p[r(P')=a]$ and $\p[r(P)=b]= \p[r(P')=b]$ for all $b\succ_j a$.

Differential privacy \cite{dwork2006our} requires a function to return similar outputs while receiving similar inputs.

\begin{definition}[\bf \boldmath Differential privacy]
	\label{def: dp}
	A function $r$ with domain $\mathcal{D}$ is $\epsilon$-differentially private ($\epsilon$-DP for short) if for all $O\subseteq \operatorname{Range}(r)$ and $P,P'\in \mathcal{D}$ differing on only one record,
	\begin{align*}
		\p[ r(P) \in O] \leqslant \erm^\epsilon\cdot \p[ r(P') \in O].
	\end{align*}
\end{definition}
In other words, a function $r$ is $\epsilon$-DP, if the ratio between the probabilities for the outputs of any pair of neighboring datasets to be in any given set $O$ must be upper bounded by $\erm^\epsilon$.
In the context of social choice, $r$ is a voting rule and
\begin{align*}
	\mathcal{D}=\LA^*=\LA\cup \LA^2\cup \cdots,
\end{align*}
and $P,P'$ are two profiles differing on only one voter's vote.

Notice that Definition~\ref{def: dp} does not require the $\erm^\epsilon$ upper bound to be tight. The tight upper bound is captured by {\em exact DP}, formally defined as follows.

\begin{definition}[\bf \boldmath Exact differential privacy \cite{Dwork06D}]
	A voting rule $r$ is exact DP ($\epsilon$-eDP for short) if it is $\epsilon$-DP and there does not exist $\epsilon'<\epsilon$ such that $r$ is $\epsilon'$-DP.
\end{definition}

For both DP and eDP, the privacy budget $\epsilon$ usually is decided according to the users' demand. For example, iOS 11 requires $\epsilon\leq43$ and iOS 10 requires $\epsilon\leq14$~\cite{orr2017google}\footnote{iOS has may have stronger privacy requirement for some specific data types (e.g., $\epsilon\leq 8$ for Safari Auto-play intent detection data) \cite{Apple}.}. In the next section, we provide upper and lower bounds for the required noise level for any user-defined privacy budget.

\section{Differentially Private Condorcet Methods}
\label{sec: mechanisms}

In this section, we propose a novel class of randomized voting rules. We apply three randomization mechanisms and obtain three classes of voting rules. By analyzing the worst cases, we prove that all of the three rules are differentially private, and our bounds of privacy budget are tight.

As mentioned in Section \ref{sec: pre}, the existence of Condorcet winner is completely determined by the UMG. In our mechanism, denoted by $\CMRand_\lambda$, a randomization mechanism \texttt{Rand} generates a noisy UMG for the given profile, and the voting rule outputs the Condorcet winner. If the Condorcet winner does not exist, the mechanism will generate another UMG, until the Condorcet winner exists, as shown in Mechanism \ref{algo: mech}.

\renewcommand{\algorithmcfname}{Mechanism}
\begin{algorithm}
	\caption{Randomized Condorcet Method}
	\label{algo: mech}
	\KwIn{Profile $P$, Parameter $\lambda$, Randomization $\texttt{Rand}$}
	\KwOut{Winning alternative}
	\SetKwFunction{select}{Select\_Rand}
	\SetKwFunction{rand}{Rand}
	\SetKwFunction{RCM}{CM\_Rand}
	\SetKwProg{Fn}{Function}{:}{}
	\Fn{\select{$S$, $\lambda$}}{
		Get randomized unweighted graph $U_{\lambda,P}^\texttt{Rand}$ with randomized mechanism \rand\;
		\uIf{There exists Condorcet winner $a$ for $U_{\lambda,P}^\texttt{Rand}$}
		{
			\KwRet $a$\;
		}
		\Else
		{
			\select{$S$,$\lambda$}\;
		}
	}
	\Fn{\RCM{$P$, $\lambda$}}{
	Compute $S_P[a,b]$ for all $a,b\in A$\;
	\select{$S_P$, $\lambda$}\;
	}
\end{algorithm}

\begin{remark}
	Notice that for each pair of alternatives $a,b\in A$, $U_{\lambda,P}^\texttt{Rand}[a,b]$ and $U_{\lambda,P}^\texttt{Rand}[b,a]$ are determined simultaneously, i.e., $U_{\lambda,P}^\texttt{Rand}[a,b]=1$, if and only if $U_{\lambda,P}^\texttt{Rand}[b,a]=-1$. Thus, any noisy UMG $U_{\lambda,P}^\texttt{Rand}$ produced in Mechanism \ref{algo: mech} claims at most one Condorcet winner. In other words, our mechanism is a well-defined map from $\LA^*$ to $\Delta(A)$.
\end{remark}

In the randomization process, we adopt three different methods, which are defined as follows.

\begin{definition}
	\label{def: 3random}
	Given $\lambda>0$, the three randomization mechanisms are
	\begin{itemize}
		\item Laplace mechanism: Given profile $P$, for any $a_{i}, a_j\in A$ that $i<j$, let $\hat{w}_P[a_i,a_j]=w_P[a_i,a_j] + X_{ij}$ for all $a_i,a_j\in A$ and $\hat{w}_P[a_j,a_i]=-\hat{w}_P[a_i,a_j]$, where $X_{ij} \overset{i.i.d}{\sim} \lap(1/\lambda)$\footnote{The Laplace distribution with scale parameter $1/\lambda$, of which the probability density function (PDF) is $f_\lambda(x)=\frac{\lambda}{2}\erm^{-\lambda|x|}$.}. Under such a mechanism, the noisy UMG is
		      \begin{align*}
			      U_{\lambda,P}^\text{\rm LAP}[a,b]=\sgn(\hat{w}_P[a,b]).
		      \end{align*}
		\item Exponential mechanism: For profile $P$,
		      \begin{align*}
			        & \p[ U_{\lambda,P}^\text{\rm EXP}[a,b]=1 ] \propto \erm^{\lambda\cdot S_P[a,b]/2}, \\
			        & \p[ U_{\lambda,P}^\text{\rm EXP}[a,b]=-1 ] \propto \erm^{\lambda
			      \cdot S_P[b,a]/2}.
		      \end{align*}
		\item Randomized response: For the majority margin $ w_P$ of a given profile $P$, if $ w_P[a,b]\neq 0$,
		      \begin{align*}
			      U_{\lambda,P}^\text{\rm RR}[a,b]=\begin{cases}
				      \sgn(w_P[a,b]),  & \text{ w.p. }\frac{\erm^{\lambda}}{1+\erm^{\lambda}}, \\
				      -\sgn(w_P[a,b]), & \text{ w.p. }\frac{1}{1+\erm^{\lambda}}.              \\
			      \end{cases}
		      \end{align*}
		      If $ w_P[a,b]=0$, then
		      \begin{align*}
			      \p[U_{\lambda,P}^\text{\rm RR}[a,b]=1]=\p[U_{\lambda,P}^\text{\rm RR}[a,b]=-1]=1/2.
		      \end{align*}
	\end{itemize}
\end{definition}
The three randomization mechanisms above are denoted by LAP, EXP, and RR, respectively. For each $\texttt{Rand}\in \RandSet$, the Condorcet winner may not exist for the noisy UMG $U^\texttt{Rand}_{\lambda,P}$. Thus, our mechanism may need to perform the randomization for several times. In fact, for any given profile $P$, the expected times of randomization is $\exp(\Theta(m))$ (see Appendix). However, such a mechanism with high time complexity can be sampled efficiently, as shown in the following lemma.

\begin{lemma}
	\label{lem: mech2scf}
	For any $\texttt{Rand}\in \RandSet$ and $\lambda>0$, $\CMRand_\lambda$ can be sampled as follows:
	\begin{itemize}
		\item For any $P\in \LA^*$, $\CMLAP_\lambda(P)$ is a probability distribution in $\Delta(A)$, such that for any $a\in A$,
		      \begin{align*}
			      \p[ \CMLAP_\lambda(P)= & a ] \propto \prod_{b\neq a}\CDF(w_P[a,b]),
		      \end{align*}
		      where $F_\lambda(x)=\int_{-\infty}^{x} f_\lambda(t)\drm t$ is the cumulative distribution function (CDF) of $\lap(1/\lambda)$.
		\item For any $P\in \LA^*$, $\CMEXP_\lambda(P)$ is a probability distribution in $\Delta(A)$, such that for any $a\in A$,
		      \begin{align*}
			      \p[ \CMEXP_\lambda(P)=a ]\propto \prod_{b\neq a}\frac{1}{1+\erm^{-\lambda\cdot w_P[a,b]/2}}.
		      \end{align*}
		\item For any $P\in \LA^*$, $\CMRR_\lambda(P)$ is a probability distribution in $\Delta(A)$, such that for any $a\in A$,
		      \begin{align*}
			      \p[ \CMRR_\lambda(P)=a ]\propto \frac{\erm^{\lambda\cdot |B(a)|}}{(1+\erm^\lambda)^{m-1}},
		      \end{align*}
		      where $B(a)=\{b\in A: S_P[a,b]>S_P[b,a]\}$.
	\end{itemize}
\end{lemma}

\begin{proof}
	Since $\CMRand_\lambda$ will keep performing the randomization on $S$ until the Condorcet winner for $U_{\lambda,P}^\texttt{Rand}$ exists, the winning probability of each $a\in A$ is determined by the conditional probability $\p[ a~\text{wins}~|~\CW(U_{\lambda,P}^\texttt{Rand})~\text{exists} ]$. First, for $\CMLAP_\lambda$, we have
	\begin{align*}
		  & \p[ a~\text{wins}~|~\CW(U_{\lambda,P}^\text{\rm LAP})~\text{exists} ] = \prod_{b\in A\backslash\{a\} }\p[U_{\lambda,P}^\text{\rm LAP}[a,b]=1] \\
		= & \prod_{b\in A\backslash\{a\} }\p[X_{ba}-X_{ab}<w_P[a,b]].
	\end{align*}
	For any $a, b\in A$, the probability density function (PDF) of $X_{ab}-X_{ba}$ is (see Appendix A for the proof)
	\begin{align*}
		\label{equ: pdf}
		\PDF(x) = \frac{\lambda+\lambda^2 |x|}{4}\cdot\erm^{-\lambda |x|}.
	\end{align*}
	Therefore, the cumulative distribution function (CDF) is
	\begin{equation*}
		\begin{aligned}
			\CDF(x) = \frac{1}{2} + \sgn(x)\cdot \left(\frac{1}{2}-\frac{2+\lambda |x|}{4}\erm^{-\lambda |x|} \right).
		\end{aligned}
	\end{equation*}
	Then we have
	\begin{align*}
		\p[ a~\text{wins}~|~\CW(U_{\lambda,P}^\text{\rm LAP})~\text{exists} ] = \prod_{b\in A\backslash\{a\}}\CDF(w_P[a,b]),
	\end{align*}
	For $\CMEXP_\lambda$, we have
	\begin{align*}
		  & \p[ a~\text{wins}~|~\CW(U_{\lambda,P}^\text{\rm EXP})~\text{exists} ] = \prod_{b\in A\backslash\{a\}}\p[U_{\lambda,P}^\text{\rm EXP}[a,b]=1]                                                                 \\
		= & \prod_{b\in A\backslash\{a\}}\frac{\erm^{\lambda\cdot S_P[a,b]/2}}{\erm^{\lambda\cdot S_P[a,b]/2}+\erm^{\lambda\cdot S_P[b,a]/2}} =\prod_{b\in A\backslash\{a\}}\frac{1}{1+\erm^{-\lambda\cdot w_P[a,b]/2}}.
	\end{align*}
	For $\CMRR_\lambda$, we have
	\begin{align*}
		  & \p[ a~\text{wins}~|~\CW(U_{\lambda,P}^\text{\rm RR})~\text{exists} ] = \prod_{b\in A\backslash\{a\}}\p[U_{\lambda,P}^\text{\rm RR}[a,b]=1]                          \\
		= & \prod_{b\in B(a)}\frac{\erm^\lambda}{1+\erm^\lambda}\cdot \prod_{b\notin B(a)}\frac{1}{1+\erm^\lambda} = \frac{\erm^{\lambda\cdot |B(a)|}}{(1+\erm^\lambda)^{m-1}},
	\end{align*}
	where $B(a)=\{b\in A: S_P[a,b]>S_P[b,a]\}$, which completes the proof.
\end{proof}

Since there are totally $m$ alternatives, and the value of $\p[ \CMRand_\lambda(P)=a ]$ for each $a\in A$ in Lemma \ref{lem: mech2scf} can be computed in $O(m)$ time, $\CMRand_\lambda$ can be sampled in $O(m^2)$ time.

Now, we are ready to show the DP bounds of our rules. For simplicity , we use $G_\lambda(x)$ to denote $\p[U_{\lambda,P}^\texttt{Rand}[a,b]=1]$, where $w_P[a,b]=x$. For example, when $\texttt{Rand} = {\rm LAP}$, $G_\lambda(x)=F_\lambda(x)$; when $\texttt{Rand} = {\rm EXP}$, $G_\lambda(x)=\frac{1}{1+\erm^{-\lambda x/2}}$.

\begin{theorem}
	\label{thm: dp}
	Given $\lambda>0$ and $\texttt{Rand}$, suppose that $\CMRand_\lambda$ satisfies $\epsilon$-eDP. When $\texttt{Rand}\in \{ {\rm LAP},{\rm EXP} \}$
	\begin{align*}
		\ln \left(\frac{G_\lambda^{m-1}(2)-G_\lambda^{m-1}(-2)}{G_\lambda(2)-G_\lambda(-2)}\cdot \frac{2^{m-2}}{m-1}\right)+(m-1)\lambda \leqslant \epsilon \leqslant 2(m-1)\lambda.
	\end{align*}
	When $\texttt{Rand}={\rm RR}$, $(m-1)\lambda\leqslant \epsilon \leqslant 2(m-1)\lambda$.
\end{theorem}

\begin{proof}
	To prove the upper bound, we only need to prove that for each $\texttt{Rand}\in \RandSet$, and neighboring profiles $P,P'$, $\frac{\p[\CMRand_\lambda(P) = a] }{\p[\CMRand(P') = a]} \leqslant \erm^{2(m-1)\lambda}$.
	W.l.o.g., we make comparison between the winning probabilities of $a_1$ for profiles $P$ and $P'$. According to Lemma \ref{lem: mech2scf}, for any $\texttt{Rand}\in \RandSet$, we have
	\begin{align*}
		\frac{\p[\CMRand_\lambda(P)=a_1]}{\p[\CMRand_\lambda(P')=a_1]}
		= & \frac{ \prod\limits_{i=2}^m \p[U_{\lambda,P}^\texttt{Rand}[a_1,a_i]=1] / \sum\limits_{j=1}^m \prod\limits_{i\neq j}\p[U_{\lambda,P}^\texttt{Rand}[a_j,a_i]=1] }{ \prod\limits_{i=2}^m \p[U_{\lambda,P'}^\texttt{Rand}[a_1,a_i]=1] / \sum\limits_{j=1}^m \prod\limits_{i\neq j}\p[U_{\lambda,P'}^\texttt{Rand}[a_j,a_i]=1] }                \\
		= & \frac{ \prod\limits_{i=2}^m \p[U_{\lambda,P}^\texttt{Rand}[a_1,a_i]=1] }{ \prod\limits_{i=2}^m \p[U_{\lambda,P'}^\texttt{Rand}[a_1,a_i]=1] } \cdot \frac{ \sum\limits_{j=1}^m \prod\limits_{i\neq j}\p[U_{\lambda,P'}^\texttt{Rand}[a_j,a_i]=1] }{ \sum\limits_{j=1}^m \prod\limits_{i\neq j}\p[U_{\lambda,P}^\texttt{Rand}[a_j,a_i]=1] }.
	\end{align*}
	When $\texttt{Rand}={\rm LAP}$, since for any $a_j\in A$,
	\begin{align*}
		  & \frac{ \prod\limits_{i\neq j} \p[U_{\lambda,P}^\text{\rm LAP}[a_j,a_i]=1] }{ \prod\limits_{i\neq j} \p[U_{\lambda,P'}^\text{\rm LAP}[a_j,a_i]=1] } \leqslant \max_{P,P'}\frac{\prod\limits_{i\neq j} \CDF(w_P[a_j,a_i])}{\prod\limits_{i\neq j} \CDF(w_{P'}[a_j,a_i])} = \max_{P,P'}\prod\limits_{i\neq j} \frac{ \CDF(w_P[a_j,a_i]) }{\CDF(w_{P'}[a_j,a_i]) } \leqslant \left( \max_{P,P'}\frac{\CDF(w_P[a_j,a_i])}{\CDF(w_{P'}[a_j,a_i])} \right)^{m-1}.
	\end{align*}
	Letting $x= w_{P'}[a_1,a_i]$, we have $ w_P[a_1,a_i]\in [x-2,x+2]$. Due to the monotonicity of CDF, we have
	\begin{align*}
		  & \max_{P,P'}\frac{\p[\CMLAP_\lambda(P)=a_1]}{\p[\CMLAP_\lambda(P')=a_1]}\leqslant \left(\max_{x\in\mathbb{Z}}\max_{|t-x|\leqslant 2}\frac{\CDF(t)}{\CDF(x)}\right)^{m-1} = \left( \max_{x\in\mathbb{Z}} \frac{ \CDF(x+2) }{ \CDF(x) } \right)^{m-1} = \erm^{(m-1)\lambda}.
	\end{align*}
	Then it follows that
	\begin{align*}
		\frac{ \sum\limits_{j=1}^m \prod\limits_{i\neq j}\p[U_{\lambda,P'}^\text{\rm LAP}[a_j,a_i]=1] }{ \sum\limits_{j=1}^m \prod\limits_{i\neq j}\p[U_{\lambda,P}^\text{\rm LAP}[a_j,a_i]=1] } \leqslant \erm^{(m-1)\lambda}.
	\end{align*}
	Further, we have
	\begin{align*}
		\frac{\p[\CMLAP_\lambda(P)=a_1]}{\p[\CMLAP_\lambda(P')=a_1]} \leqslant \erm^{(m-1)\lambda}\cdot \erm^{(m-1)\lambda} = \erm^{2(m-1)\lambda}.
	\end{align*}

	When $\texttt{Rand}={\rm EXP}$, for any profile $P\in \LA^n$, we have
	\begin{align*}
		\p[U_{\lambda,P}^\text{\rm EXP}[a_1,a_i]=1] = \frac{\erm^{\lambda\cdot w_P[a_1,a_i]/2}}{1+ \erm^{\lambda\cdot w_P[a_1,a_i]/2}},
	\end{align*}
	which indicates that
	\begin{align*}
		\frac{\prod\limits_{i=2}^m \p[U_{\lambda,P}^\text{\rm EXP}[a_1,a_i]=1]}{\prod\limits_{i=2}^m \p[U_{\lambda,P'}^\text{\rm EXP}[a_1,a_i]=1]} & \leqslant \max_{P,P'} \prod_{i=2}^m \frac{\erm^{\lambda\cdot w_P[a_1,a_i]/2}}{1+\erm^{\lambda\cdot w_P[a_1,a_i]/2}}\cdot \frac{1+\erm^{\lambda\cdot w_{P'}[a_1,a_i]/2}}{\erm^{\lambda\cdot w_{P'}[a_1,a_i]/2}}                                                                                          \\
		                                                                                                                                           & =          \prod_{i=2}^m \max_{x\in\mathbb{Z}}\frac{\erm^{\lambda (x+2)/2}}{1+\erm^{\lambda (x+2)/2}} \cdot \frac{1+\erm^{\lambda\cdot x/2}}{\erm^{\lambda\cdot x/2}} =          \left( \max_{x\in \mathbb{Z}} \frac{\erm^{\lambda}(1+\erm^{\lambda\cdot x/2})}{1+\erm^{\lambda (2+x)/2}} \right)^{m-1} \\
		                                                                                                                                           & \leqslant  \left( \lim\limits_{x\to -\infty} \frac{\erm^{\lambda}(1+\erm^{\lambda\cdot x/2})}{1+\erm^{\lambda (2+x)/2}} \right)^{m-1} = \erm^{(m-1)\lambda}.
	\end{align*}
	Further, we have
	\begin{align*}
		\frac{ \sum\limits_{j=1}^m \prod\limits_{i\neq j}\p[U_{\lambda,P'}^\text{\rm EXP}[a_j,a_i]=1] }{ \sum\limits_{j=1}^m \prod\limits_{i\neq j}\p[U_{\lambda,P}^\text{\rm EXP}[a_j,a_i]=1] } \leqslant \erm^{(m-1)\lambda}.
	\end{align*}
	As a consequence,
	\begin{align*}
		\frac{\p[\CMEXP_\lambda(P)=a_1]}{\p[\CMEXP_\lambda(P')=a_1]} = \frac{ \prod\limits_{i=2}^m \p[U_{\lambda,P}^\text{\rm EXP}[a_1,a_i]=1] }{ \prod\limits_{i=2}^m \p[U_{\lambda,P'}^\text{\rm EXP}[a_1,a_i]=1] } \cdot \frac{ \sum\limits_{j=1}^m \prod\limits_{i\neq j}\p[U_{\lambda,P'}^\text{\rm EXP}[a_j,a_i]=1] }{ \sum\limits_{j=1}^m \prod\limits_{i\neq j}\p[U_{\lambda,P}^\text{\rm EXP}[a_j,a_i]=1] } \leqslant \erm^{2(m-1)\lambda}.
	\end{align*}

	Finally, when $\texttt{Rand}={\rm RR}$, for any profile $P\in \LA^n$, we have
	\begin{align*}
		\prod_{i=2}^m \p[U_{\lambda,P}^\text{\rm RR}[a_1,a_i]=1] = \frac{\erm^{\lambda |\{a_i\in A: U_P[a_1,a_i]>0\}|}}{(1+\erm^{\lambda})^{m-1}}.
	\end{align*}
	Thus, for any neighboring profiles $P,P'$,
	\begin{align*}
		  & \frac{\prod\limits_{i=2}^m \p[U_{\lambda,P}^\text{\rm RR}[a_1,a_i]=1]}{\prod\limits_{i=2}^m \p[U_{\lambda,P}^\text{\rm RR}[a_1,a_i]=1]} = \frac{\erm^{\lambda |\{a_i\in A: U_P[a_1,a_i]>0\}|}}{\erm^{\lambda |\{a_i\in A: U_{P'}[a_1,a_i]>0\}|}}\leqslant \frac{\max\limits_{P\in\LA^n}\erm^{\lambda |\{a_i\in A: U_P[a_1,a_i]>0\}|}}{\min\limits_{P'\in\LA^n}\erm^{\lambda |\{a_i\in A: U_{P'}[a_1,a_i]>0\}|}} =\erm^{(m-1)\lambda}.
	\end{align*}
	Further,
	\begin{align*}
		\frac{\sum\limits_{j=1}^m\prod\limits_{i\neq j} \p[U_{\lambda,P}^\text{\rm RR}[a_1,a_i]=1]}{\sum\limits_{j=1}^m\prod\limits_{i\neq j} \p[U_{\lambda,P}^\text{\rm RR}[a_1,a_i]=1]} \leqslant \frac{\sum\limits_{j=1}^m \erm^{(m-1)\lambda}\cdot \prod\limits_{i\neq j} \p[U_{\lambda,P}^\text{\rm RR}[a_1,a_i]=1]}{\sum\limits_{j=1}^m\prod\limits_{i\neq j} \p[U_{\lambda,P}^\text{\rm RR}[a_1,a_i]=1]} = \erm^{2(m-2)\lambda}.
	\end{align*}
	Then we have
	\begin{align*}
		\frac{\p[\CMRR_\lambda(P)=a_1]}{\p[\CMRR_\lambda(P')=a_1]} = \frac{ \prod\limits_{i=2}^m \p[U_{\lambda,P}^\text{\rm RR}[a_1,a_i]=1] }{ \prod\limits_{i=2}^m \p[U_{\lambda,P'}^\text{\rm RR}[a_1,a_i]=1] } \cdot \frac{ \sum\limits_{j=1}^m \prod\limits_{i\neq j}\p[U_{\lambda,P'}^\text{\rm RR}[a_j,a_i]=1] }{ \sum\limits_{j=1}^m \prod\limits_{i\neq j}\p[U_{\lambda,P}^\text{\rm RR}[a_j,a_i]=1] } \leqslant \erm^{2(m-1)\lambda},
	\end{align*}
	which completes the proof for the upper bound.

	For the lower bounds of $\CMLAP_\lambda$ and $\CMEXP_\lambda$, we only need to show that there exists neighboring profiles $P,P'$, and alternative $a$,
	\begin{align*}
		\frac{\p[\CMRand(P)=a]}{\p[\CMRand(P')=a]} \geqslant  \frac{G_\lambda^{m-1}(2)-G_\lambda^{m-1}(-2)}{G_\lambda(2)-G_\lambda(-2)}\cdot \frac{2^{m-2}}{m-1} \cdot \erm^{(m-1)\lambda}.
	\end{align*}
	Consider the following profile $P$ (let $m=2k$):
	\begin{itemize}
		\item $k$ voters: $a_1\succ a_2\succ\cdots\succ a_m$;
		\item $k-1$ voters: $a_{m-1}\succ a_{m-2}\succ \cdots a_1 \succ a_m$;
		\item $1$ voter: $a_m\succ a_{m-1} \succ \cdots \succ a_1$.
	\end{itemize}
	And another profile $P'$:
	\begin{itemize}
		\item $k+1$ voters: $a_1\succ' a_2\succ'\cdots\succ' a_m$;
		\item $k-1$ voters: $a_{m-1}\succ' a_{m-2}\succ' \cdots a_1 \succ' a_m$;
	\end{itemize}
	It is quite easy to verify that $P$ and $P'$ are neighboring datasets. Further, we have
	\begin{align*}
		w_P[a_i,a_j] = \begin{cases}
			0,   & 1\leqslant i,j \leqslant m-1; \\
			n-2, & i<m,j=m;                      \\
			2-n, & i=m,j<m.
		\end{cases} \qquad w_{P'}[a_i,a_j] = \begin{cases}
			2,  & 1\leqslant i < j \leqslant m-1; \\
			-2, & 1\leqslant j < i \leqslant m-1; \\
			n   & i<m, j=m;                       \\
			-n  & i=m, j<m.
		\end{cases}
	\end{align*}
	Therefore,
	\begin{align*}
		\frac{\sum\limits_{i=1}^m \prod\limits_{j\neq i} \p[U_{\lambda,P'}^\texttt{Rand}[a_i,a_j]=1]}{\sum\limits_{i=1}^m \prod\limits_{j\neq i} \p[U_{\lambda,P}^\texttt{Rand}[a_i,a_j]=1]} = \frac{G_\lambda(n-2) \sum\limits_{k=0}^{m-2} G_\lambda^k(2)G_\lambda^{m-k-2}(-2) + G_\lambda^{m-1}(2-n)}{(m-1)G_\lambda(n)G_\lambda^{m-2}(0) + G_\lambda^{m-1}(-n)}
	\end{align*}
	When $n\to +\infty$, we have $G(n)\to 1$ and $G(-n)\to 0$. Then
	\begin{align*}
		\lim\limits_{n\to+\infty}\frac{\sum\limits_{i=1}^m \prod\limits_{j\neq i} \p[U_{\lambda,P'}^\texttt{Rand}[a_i,a_j]=1]}{\sum\limits_{i=1}^m \prod\limits_{j\neq i} \p[U_{\lambda,P}^\texttt{Rand}[a_i,a_j]=1]} = \frac{\sum\limits_{k=0}^{m-2} G_\lambda^k(2)G_\lambda^{m-k-2}(-2)}{(m-1)G_\lambda^{m-2}(0)} = \frac{G_\lambda^{m-1}(2)-G_\lambda^{m-1}(-2)}{G_\lambda(2)-G_\lambda(-2)}\cdot \frac{2^{m-2}}{m-1}.
	\end{align*}
	In other words, when $\texttt{Rand} \in \{ {\rm LAP}, {\rm EXP} \}$,
	\begin{align*}
		\frac{\sum\limits_{i=1}^m \prod\limits_{j\neq i} \p[U_{\lambda,P'}^\texttt{Rand}[a_i,a_j]=1]}{\sum\limits_{i=1}^m \prod\limits_{j\neq i} \p[U_{\lambda,P}^\texttt{Rand}[a_i,a_j]=1]} \geqslant \frac{G_\lambda^{m-1}(2)-G_\lambda^{m-1}(-2)}{G_\lambda(2)-G_\lambda(-2)}\cdot \frac{2^{m-2}}{m-1}.
	\end{align*}
	Then we have
	\begin{align*}
		\frac{\p[\CMRand(P')=a_m]}{\p[\CMRand(P)=a_m]} & = \frac{\prod\limits_{j\neq i} \p[U_{\lambda,P}^\texttt{Rand}[a_i,a_j]=1]}{\prod\limits_{j\neq i} \p[U_{\lambda,P'}^\texttt{Rand}[a_i,a_j]=1]} \cdot \frac{\sum\limits_{i=1}^m \prod\limits_{j\neq i} \p[U_{\lambda,P'}^\texttt{Rand}[a_i,a_j]=1]}{\sum\limits_{i=1}^m \prod\limits_{j\neq i} \p[U_{\lambda,P}^\texttt{Rand}[a_i,a_j]=1]} \\
		                                               & \geqslant \frac{G_\lambda^{m-1}(2)-G_\lambda^{m-1}(-2)}{G_\lambda(2)-G_\lambda(-2)}\cdot \frac{2^{m-2}}{m-1} \cdot \frac{\prod\limits_{j\neq i} \p[U_{\lambda,P}^\texttt{Rand}[a_i,a_j]=1]}{\prod\limits_{j\neq i} \p[U_{\lambda,P'}^\texttt{Rand}[a_i,a_j]=1]}                                                                           \\
		                                               & \geqslant \frac{G_\lambda^{m-1}(2)-G_\lambda^{m-1}(-2)}{G_\lambda(2)-G_\lambda(-2)}\cdot \frac{2^{m-2}}{m-1} \cdot \left( \lim\limits_{n\to +\infty} \frac{G_\lambda(2-n)}{G_\lambda(-n)} \right)^{m-1}                                                                                                                                   \\
		                                               & = \frac{G_\lambda^{m-1}(2)-G_\lambda^{m-1}(-2)}{G_\lambda(2)-G_\lambda(-2)}\cdot \frac{2^{m-2}}{m-1} \cdot \erm^{(m-1)\lambda}.
	\end{align*}

	Finally, for $\CMRR_\lambda$, consider the following profile (let $m=2k+1$):
	\begin{itemize}
		\item $k+1$ voters: $a_1 \succ a_2 \succ \cdots \succ a_m$;
		\item $k$ voters: $a_m \succ' a_{m-1} \succ' \cdots \succ' a_1$.
	\end{itemize}

	And its neighboring profile $P'$:
	\begin{itemize}
		\item $k$ voters: $a_1 \succ a_2 \succ \cdots \succ a_m$;
		\item $k+1$ voters: $a_m \succ' a_{m-1} \succ' \cdots \succ' a_1$.
	\end{itemize}

	Then the majority margin of $P$ and $P'$ should be

	\begin{align*}
		w_P[a_i,a_j] = \begin{cases}
			1,  & i < j, \\
			-1, & i > j.
		\end{cases} \qquad w_{P'}[a_i,a_j] = \begin{cases}
			1,  & i < j, \\
			-1, & i > j.
		\end{cases}
	\end{align*}

	Further, we have

	\begin{align*}
		\frac{\p[\CMRand(P')=a_m]}{\p[\CMRand(P)=a_m]} & = \frac{\prod\limits_{j\neq i} \p[U_{\lambda,P}^\texttt{Rand}[a_i,a_j]=1]}{\prod\limits_{j\neq i} \p[U_{\lambda,P'}^\texttt{Rand}[a_i,a_j]=1]} \cdot \frac{\sum\limits_{i=1}^m \prod\limits_{j\neq i} \p[U_{\lambda,P'}^\texttt{Rand}[a_i,a_j]=1]}{\sum\limits_{i=1}^m \prod\limits_{j\neq i} \p[U_{\lambda,P}^\texttt{Rand}[a_i,a_j]=1]} = \erm^{(m-1)\lambda},
	\end{align*}
	which completes the proof.
\end{proof}

With Theorem \ref{thm: dp}, we can get the upper and lower bounds of noise level $\lambda$ for any given $\epsilon$. The relations between the lower and upper bounds when $m=5$ and $m=20$ are shown in Figure \ref{fig: DP-value}.

\begin{figure}[H]
	\centering
	\includegraphics[width=.45\linewidth]{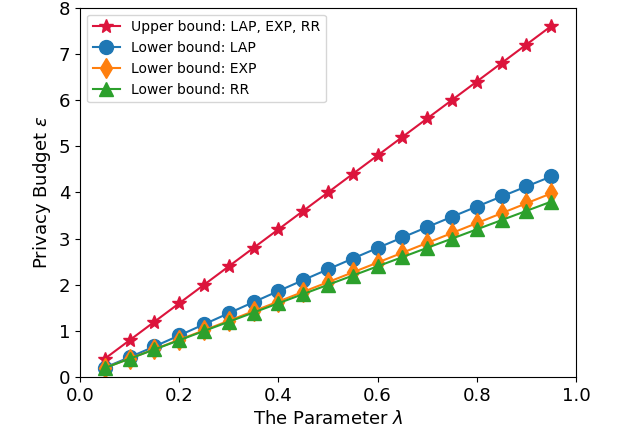}
	\includegraphics[width=.45\linewidth]{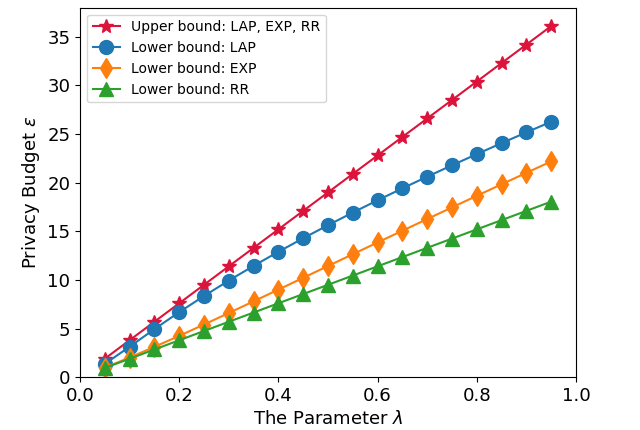}
	\caption{The lower and upper bounds of privacy budget (left: $m=5$, right: $m=20$).}
	\label{fig: DP-value}
\end{figure}

\section{Axioms-Privacy Tradeoff}
\label{sec: axioms}

In this section, we analyze our voting rules with axioms mentioned in Section~\ref{sec: pre}. We show that our rules do not satisfy Condorcet criterion and Pareto efficiency. To address these challenges, we explore probabilistic variants of them. Then we discuss the satisfaction to absolute monotonicity, SD-strategyproofness, and lexi-participation.

To begin with, we analyze our voting rules with Condorcet criterion. But unfortunately, $\CMRand_\lambda$ does not satisfy Condorcet criterion with any $\lambda$ and $\texttt{Rand}\in\RandSet$, though it is based on the Condorcet method. Intuitively, for any $P\in\LA^n$ and $a,b\in A$, $\p[U^\texttt{Rand}_{\lambda,P}[a,b]=1]<1$. Then
\begin{align*}
	\p[\CMRand_\lambda(P)=a]\leqslant \prod_{b\in A\backslash\{a\}} \p[U^\texttt{Rand}_{\lambda,P}[a,b]=1]<1,
\end{align*}
even when $a$ is the Condorcet winner. To deal with this, we propose a probabilistic variant of Condorcet criterion, which is shown in the following definition.

\begin{definition}[\bf Probabilistic Condorcet criterion]
	\label{def: p-cond}
	A randomized voting rule $r$ satisfies probabilistic Condorcet criterion (p-Condorcet) if for every profile $P$ that $\CW(P)$ exists and all $a\in A\backslash\{ \CW (P)\}$,
	\begin{align*}
		\p[r(P)=\CW (P)] \geqslant \p[r(P)=a].
	\end{align*}
\end{definition}

At a high level, Definition \ref{def: p-cond} is a relaxation of the Condorcet criterion, since it does not always require the voting rule $r$ to select the Condorcet winner. Further, the following theorem holds. The proof can be found in the full version.

\begin{theorem}
	\label{thm: sat-pCond}
	For any $\lambda>0$, $\CMRR_\lambda$ satisfies p-Condorcet.
\end{theorem}

\begin{proof}
	Since there are only one profile throughout the proof, we omit the normalization here, i.e., for profile $P$, supposing $\CW (P)=a$, we have
	$$\p[\CMRR_\lambda(P)=a] = \left( \frac{\erm^\lambda}{\erm^\lambda+1} \right)^{m-1}.$$
	For any $b\in A \backslash \{a\}$, since $b$ is not the Condorcet winner, there exists a nonempty set $B\subseteq A$, such that $C[c,b]=1$ if and only if $c\in B$. Hence,
	\begin{align*}
		\p[\CMRR_\lambda(P)=b] = \frac{\erm^{(m-|B|-1)\lambda}}{(\erm^\lambda+1)^{m-1}} \leqslant \p[\CMRR_\lambda(P)=a],
	\end{align*}
	which completes the proof.
\end{proof}

In other words, $\CMRR_\lambda$ satisfies a weak version of Condorcet criterion for any $\lambda>0$. However, the results for $\CMEXP_\lambda$ and $\CMLAP_\lambda$ are relatively negative.

\begin{proposition}
	\label{prop: unsat-pCond}
	$\CMEXP_{0.5}$ and $\CMLAP_{0.5}$ do not satisfy p-Condorcet.
\end{proposition}

\begin{proof}
	For $\CMLAP_{0.5}$, consider the voting instance where $n=101$ voters and $m=5$ alternatives. The ballots are as follows
	\begin{itemize}
		\item $51$ voters: $a_1\succ a_2\succ a_3\succ a_4\succ a_5$;
		\item $50$ voters: $a_2\succ a_3\succ a_4\succ a_5\succ a_1$;
	\end{itemize}
	Then we have
	\begin{equation*}
		\begin{aligned}
			S_P[a_1,a_i]=51, S_P[a_i,a_1]=50, & \text{ for all } a_i\in A\backslash\{a_1\},     \\
			S_P[a_2,a_i]=101, S_P[a_i,a_2]=0, & \text{ for all } a_i\in A\backslash\{a_1,a_2\}.
		\end{aligned}
	\end{equation*}
	In other words,
	\begin{align*}
		w_P[a_1,a_i]=1,   & \text{ for all } a_i\in A\backslash\{a_1\},     \\
		w_P[a_2,a_i]=101, & \text{ for all } a_i\in A\backslash\{a_1,a_2\}.
	\end{align*}
	As a result, $\CW (P)=a_1$. Since $\lambda=1$, we have
	\begin{align*}
		  & \p[ \CMLAP_{0.5}(P)=a_1 ]  =\prod_{i=2}^5 \CDF(1) = \frac{1}{16}\left( 2-\frac{1}{\sqrt{\erm}}\right)^4 \approx 0.2357.
	\end{align*}
	However,
	\begin{align*}
		  & \p[ \CMLAP_{0.5}(P)=a_2 ]  = \CDF(-1)\cdot \CDF^3(101) = \frac{1}{16\sqrt{\erm}}\left( 2-\frac{1}{\erm^{101/2}} \right)^3 \approx 0.3033.
	\end{align*}
	In other words,
	\begin{align*}
		\p[ \CMLAP_{0.5}(P)=a_1 ] < \p[ \CMLAP_{0.5}(P)=a_2 ],
	\end{align*}
	which indicates that $\CMLAP_0.5$ does not satisfy p-Condorcet.

	For $\CMEXP_{0.5}$, let $n=101$, $m=5$, and the profile $P$ be
	\begin{itemize}
		\item $51$ voters: $a_1\succ a_2\succ a_3\succ a_4\succ a_5$;
		\item $50$ voters: $a_2\succ a_3\succ a_4\succ a_5\succ a_1$.
	\end{itemize}
	Then we have
	\begin{equation*}
		\begin{aligned}
			  & S_P[a_1,a_i]=51,  &   & S_P[a_i,a_1]=50, &   & \text{ for all } a_i\in A\backslash\{a_1\},     \\
			  & S_P[a_2,a_i]=101, &   & S_P[a_i,a_2]=0,  &   & \text{ for all } a_i\in A\backslash\{a_1,a_2\}.
		\end{aligned}
	\end{equation*}
	In other words,
	\begin{equation*}
		\begin{aligned}
			  & w_P[a_1,a_i]=1,   &   & \text{ for all } a_i\in A\backslash\{a_1\},     \\
			  & w_P[a_2,a_i]=101, &   & \text{ for all } a_i\in A\backslash\{a_1,a_2\}.
		\end{aligned}
	\end{equation*}
	As a result, $\CW (P)=a_1$ and
	\begin{align*}
		  & \p[ \CMEXP_{0.5}(P)=a_1 ]  =\prod_{i=2}^5\p[ U_{\lambda,P}^\text{\rm EXP}[a_1,a_i]=1 ]         \\
		= & \prod_{i=2}^5 \frac{1}{1+\erm^{-w_P[a_1,a_i]/4}} = \frac{1}{(1+\erm^{-1/4})^4} \approx 0.0999.
	\end{align*}
	However,
	\begin{align*}
		  & \p[ \CMEXP_{0.5}(P)=a_2 ]  =\prod_{i\neq 2}\p[ U_{\lambda,P}^\text{\rm EXP}[a_2,a_i]=1 ] = \frac{1}{(1+\erm^{1/4})(1+\erm^{-101/4})^3} \approx 0.4378.
	\end{align*}
	As a consequence,
	\begin{align*}
		\p[ \CMEXP_{0.5}(P)=a_1 ] < \p[ \CMEXP_{0.5}(P)=a_2 ],
	\end{align*}
	which completes the proof.
\end{proof}

To measure how much $\CMEXP_\lambda$ and $\CMLAP_\lambda$ deviate from p-Condorcet, we further extend the axiom.

\begin{definition}[\bf \boldmath $\alpha$-Probabilistic Condorcet criterion]
	\label{def: alpha-pCond}
	A randomized voting rule $r$ satisfies $\alpha$-probabilistic Condorcet criterion ($\alpha$-p-Condorcet) if for every profile $P$ that $\CW(P)$ exists and for all $a\in A\backslash \{\CW(P)\}$,
	$$\p[r(P)=\CW (P)]\geqslant \alpha \cdot\p[r(P)=a].$$
\end{definition}

Note that a larger $\alpha$ is more desirable, as $\alpha$-p-Condorcet is almost equivalent to the standard Condorcet criterion when $\alpha\to \infty$. Especially, it reduces to p-Condorcet when $\alpha=1$.

For $\CMEXP_\lambda$ and $\CMLAP_\lambda$, the following theorem holds.

\begin{theorem}
	\label{thm: sat-asympCond}
	For any $\lambda>0$,
	\begin{itemize}
		\item $\CMEXP_\lambda$ satisfies $\frac{ 1+\erm^{\lambda/2} }{ \left( 1+\erm^{-\lambda/2} \right)^{m-1} }$-p-Condorcet;
		\item $\CMLAP_\lambda$ satisfies $2\erm^{\lambda}\left(1-\frac{\erm^{-\lambda}}{2}\right)^{m-1}$-p-Condorcet;
		\item $\CMRR_\lambda$ satisfies $\erm^{\lambda}$-p-Condorcet.
	\end{itemize}
\end{theorem}

\begin{proof}
	Since there is only one profile $P$, the normalization factors in Lemma \ref{lem: mech2scf} are not necessarily considered anymore. In other words, we suppose for the sake of simplicity that
	\begin{align*}
		\p[\CMRand_\lambda(P)=a]=\prod\limits_{b\in A\backslash \{a\}} \p[U_{\lambda,P}^\texttt{Rand}[a,b]=1].
	\end{align*}

	First, we prove for $\CMEXP_\lambda$. For any $P$, letting $\CW (P)=a$, we have
	\begin{align*}
		\p[\CMEXP_\lambda(P)=a] = \prod_{c\in A\backslash \{a\}} \frac{1}{1+\erm^{\lambda\cdot w_P[c,a] /2 }}.
	\end{align*}
	Note that $\CW (P)=a$. Then $ w_P[a,c]\geqslant 1$ holds for any $c\in A\backslash\{a\}$. Hence
	\begin{align*}
		\p[\CMEXP_\lambda(P)=a] & \geqslant \prod_{c\in A\backslash \{a\}} \frac{1}{1+\erm^{ -\lambda/2 }} = \frac{1}{(1+\erm^{-\lambda/2})^{m-1}}.
	\end{align*}
	On the other hand, for any $b\in A\backslash\{a\}$, we have
	\begin{align*}
		    & \p[\CMEXP_\lambda(P)=b] = \p[U_{\lambda,P}^\text{\rm EXP}[b,a]=1] \cdot \prod_{ c\in A\backslash\{a,b\} } \p[U_{\lambda,P}^\text{\rm EXP}[b,c]=1] \\
		=\; & \frac{1}{1+\erm^{\lambda\cdot w_P[a,b] /2}}\cdot \prod_{ c\in A\backslash\{a,b\} } \frac{1}{1+\erm^{\lambda\cdot w_P[c,b]/2 }}.
	\end{align*}
	Since $b$ is not the Condorcet winner, we have $ w_P[a,b]\geqslant 1$ and $ w_P[c,b]\geqslant -n$ holds for any $a_k\in A\backslash \{a,b\}$. Therefore,
	\begin{align*}
		\p[\CMEXP_\lambda(P)=a] & \leqslant \frac{1}{1+\erm^{ \lambda/2 }} \cdot \left( \frac{1}{1+\erm^{-n\lambda/2}} \right)^{m-2} \leqslant \frac{1}{1+\erm^{ \lambda/2 }}.
	\end{align*}
	Then, for every $b\neq a$, we have
	\begin{align*}
		\frac{ \p[\CMEXP_\lambda(P)=a] }{ \p[\CMEXP_\lambda(P)=b] } \geqslant \frac{ 1+\erm^{\lambda/2} }{ \left( 1+\erm^{-\lambda/2} \right)^{m-1} }.
	\end{align*}
	That is, $\CMEXP_\lambda$ satisfies $\frac{ 1+\erm^{\lambda/2} }{ \left( 1+\erm^{-\lambda/2} \right)^{m-1} }$-p-Condorcet.

	Next, we prove for $\CMLAP_\lambda$. For any profile $P\in \LA^n$, supposing that $\CW (P)=a$, we have
	\begin{align*}
		\p[\CMLAP_\lambda(P)=a] & = \prod_{c\in A\backslash \{a\}} \p[\hat{w}_P[a,c]>0] = \prod_{c\in A\backslash \{a\}} \CDF(w_P[a,c]).
	\end{align*}
	Since $\CW (P)=a$, $S[a,c]\geqslant 1$ for any $c\in A\backslash\{a\}$. Then we have
	\begin{align*}
		\p[\CMLAP_\lambda(P)=a] & \geqslant \prod_{c\in A\backslash \{a\}} \CDF(1) = \left( 1-\frac{\erm^{-\lambda}}{2} \right)^{m-1}.
	\end{align*}
	On the other hand, for any $b\in A\backslash \{a\}$, we have
	\begin{align*}
		\p[\CMLAP_\lambda(P)=b] & = \prod_{c\in A\backslash \{a\}} \CDF(w_P[a,c]).
	\end{align*}
	Since the Condorcet winner of $P$ is $a
		_i$, $ w_P[b,a]\leqslant -1$ and $ w_P[b,c]\leqslant n$ holds for any $a_k\in A\backslash \{a,b\}$. Therefore,
	\begin{align*}
		  & \p[\CMLAP_\lambda(P)=b]
		\leqslant \CDF(-1) \CDF^{m-2}(n) \leqslant \frac{\erm^{-\lambda}}{2}.
	\end{align*}
	As a result,
	\begin{align*}
		\frac{ \p[\CMLAP_\lambda(P)=a] }{ \p[\CMLAP_\lambda(P)=b] } \geqslant 2\erm^{\lambda}\left(1-\frac{\erm^{-\lambda}}{2}\right)^{m-1}, \text{ for any } b\in A\backslash \{a\},
	\end{align*}
	which means $\CMLAP_\lambda$ satisfies $2\erm^{\lambda}\left(1-\frac{\erm^{-\lambda}}{2}\right)^{m-1}$-p-Condorcet.

	Finally, for $\CMRR_\lambda$ and any profile $P$, let $\CW(P)=a$. Then, for any $b\neq a$,
	\begin{align*}
		\frac{\p[\CMRR_\lambda(P)=a]}{\p[\CMRR_\lambda(P)=b]} = \frac{ \erm^{\lambda \cdot |\{c\in A: w_P[a,c]>0\}|} }{ \erm^{\lambda \cdot |\{c\in A: w_P[b,c]>0\}|} } = \frac{\erm^{(m-1)\lambda}}{\erm^{\lambda \cdot |\{c\in A: w_P[b,c]>0\}|}}.
	\end{align*}
	Since $b$ is not the Condorcet winner, we have $|\{c\in A: w_P[b,c]>0\}|\leqslant m-2$. Then
	\begin{align*}
		\frac{\p[\CMRR_\lambda(P)=a]}{\p[\CMRR_\lambda(P)=b]} \geqslant \frac{\erm^{(m-1)\lambda}}{\erm^{(m-2)\lambda}} = \erm^{\lambda},
	\end{align*}
	which completes the proof.
\end{proof}

With Theorem \ref{thm: sat-asympCond}, we can obtain a more general version of Proposition \ref{prop: unsat-pCond}, of which the proof is shown in the full version.

\begin{proposition}
	\label{prop: lambda-pCond}
	Given $\lambda>0$, $\CMEXP_\lambda$ satisfies p-Condorcet when $\frac{\ln(\erm^{\lambda/2}+1)}{\ln(\erm^{\lambda/2}+1)-\lambda/2}+1 \geqslant m$; $\CMLAP_\lambda$ satisfies p-Condorcet when $\frac{ \lambda+\ln 2 }{ \ln 2-\ln(2-\erm^{-\lambda}) }+1\geqslant m$.
\end{proposition}

\begin{proof}
	Let $P$ be a profile, of which the Condorcet winner $\CW(P)$ exists. First, we prove for $\CMEXP_\lambda$. Since $m\leqslant \frac{\ln(\erm^{\lambda/2}+1)}{\ln(\erm^{\lambda/2}+1)-\lambda/2}+1$, we have

	\begin{align*}
		m-1\leqslant \frac{\ln(\erm^{\lambda/2}+1)}{\ln(\erm^{\lambda/2}+1)-\lambda/2} = \log_{1+\erm^{-\lambda/2}} \left(1+\erm^{\lambda/2}\right).
	\end{align*}
	Therefore, for all $a\neq \CW(P)$,
	\begin{align*}
		\p[\CMEXP_\lambda(P)=\CW(P)] & = \prod\limits_{b\neq \CW(P)}\frac{1}{1+\erm^{-\lambda\cdot w_P[\CW(P), b]/2}} \geqslant \left( \frac{1}{1+\erm^{-\lambda/2}} \right)^{m-1}  \\
		                             & \geqslant \frac{1}{1+\erm^{\lambda/2}} \geqslant \prod\limits_{j\neq i} \frac{1}{1+\erm^{-\lambda\cdot w_P[a,b]/2}}=\p[\CMEXP_\lambda(P)=a],
	\end{align*}
	which means that $\CMEXP_\lambda$ satisfies p-Condorcet.

	Next, we prove for $\CMLAP_\lambda$. When $m\leqslant \frac{ \lambda+\ln 2 }{ \ln 2-\ln(2-\erm^{-\lambda}) }+1$, we have

	\begin{align*}
		m-1\leqslant \frac{ \lambda+\ln 2 }{ \ln 2-\ln(2-\erm^{-\lambda}) } = \log_{1-\erm^{-\lambda}/2} \left(\frac{\erm^{-\lambda}}{2}\right),
	\end{align*}
	which indicates that for any $a\neq \CW(P)$,
	\begin{align*}
		\frac{\p[\CMLAP_\lambda(P)=\CW(P)]}{\p[\CMLAP_\lambda(P)=a]} \geqslant 2\erm^{\lambda}\left(1-\frac{\erm^{-\lambda}}{2}\right)^{m-1} \geqslant 1.
	\end{align*}
	In other words, $\CMLAP_\lambda$ satisfies p-Condorcet in this case, which completes the proof.
\end{proof}

Notice that both of the LHS of the two inequalities in Proposition \ref{prop: lambda-pCond} are increasing functions of $\lambda$ which diverge when $\lambda\to \infty$. Thus, for any $m$, there must exist some $\lambda$ satisfying the inequalities.

Since the upper and lower bounds of the privacy budget can completely be determined by $\lambda$, we use $\lambda$ to denote the privacy level. Also, we use the parameter $\alpha$ in Definition \ref{def: p-cond} to denote the level of satisfaction to p-Condorcet, then the tradeoff curves when $m=5$ are shown in Figure \ref{fig: pCond-eps}.

\begin{figure}[h]
	\centering
	\includegraphics[width=.5\linewidth]{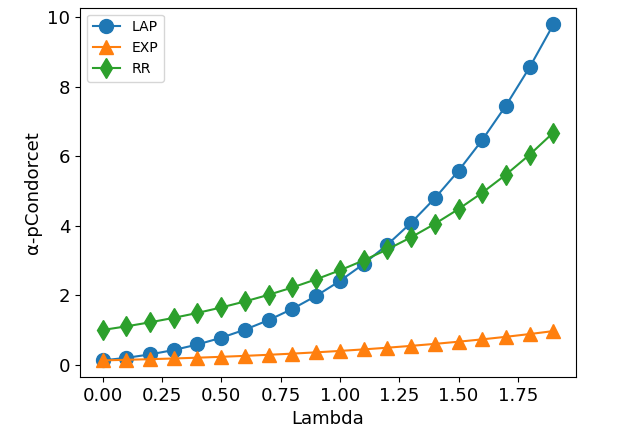}
	\caption{The satisfaction of p-Condorcet with different $\lambda$.}
	\label{fig: pCond-eps}
\end{figure}

Similar to the Condorcet criterion, our new class of voting rules do not satisfy Pareto efficiency either. Suppose there is an alternative $b\in A$, which is Pareto dominated by $a\in A$ in profile $P$, i.e., $a\succ_j b$ for all $j\in N$. Then for any $\lambda$ and $\texttt{Rand}\in\RandSet$, we have
\begin{align*}
	\p[\CMRand_\lambda(P)=b] \leqslant \prod\limits_{c\neq b}\p[U^\texttt{Rand}_{\lambda,P}[b,c]=1].
\end{align*}
According to the Definition of $\text{LAP}$, $\text{EXP}$, and $\text{RR}$, $\p[U^\texttt{Rand}_{\lambda,P}(b,c)]>0$ for all $c\in A$, which indicates that $\CMRand_\lambda$ does not satisfy Pareto efficiency. However, $a$ still dominates $b$ in another way. Formally, we have the following definition.

\begin{definition}[\bf \boldmath Probabilistic Pareto efficiency]
	\label{def: p-PO}
	A randomized voting rule $r$ satisfies probabilistic Pareto efficiency (p-Pareto) if for each pair of alternatives $a,b\in A$ that $a\succ_k b$ holds for all $k\in N$,
	$$\p[r(P)=a] \geqslant \p[r(P)=b].$$
\end{definition}

This definition is a relaxation of Pareto efficiency. For our voting rules, the following theorem holds.

\begin{theorem}
	\label{thm: sat-ppareto}
	For any $\lambda>0$, $\CMRR_\lambda$, $\CMEXP_\lambda$, and $\CMLAP_\lambda$ satisfy p-Pareto.
\end{theorem}

\begin{proof}
	Let $a,b \in A$ be the pair of alternatives that $a\succ_j b$ holds for any $j\in N$. Then for any $j\in N$ and any $c\in A\backslash\{a,b\}$ such that $b\succ_j c$, we have $a\succ_j c$. Further, we have
	\begin{equation}
		\begin{aligned}
			S[a,c] & = |\{j\in N: a\succ_j c \}| \geqslant |\{j\in N: b\succ_j c\}| = S[b, c].
		\end{aligned}
		\label{equ: S-relation}
	\end{equation}
	Hence, $ w_P[a,c] \geqslant  w_P[b,c]$, for all $c\in A\backslash \{a,b\}$. Now, we are ready to give the proof.

	For $\CMRR_\lambda$, let $B(c)=\{b\in A: U_P[c,b]=1\}$ for every $c\in A$. Then Equation (\ref{equ: S-relation}) indicates that $B(b)\subseteq B(a)$, i.e., $|B(b)|\leqslant |B(a)|$. Hence

	\begin{align*}
		          & \p[ \CMRR_\lambda(P)=a ]=\prod_{c\in A\backslash\{a\} } \p[ U_{\lambda,P}^\text{\rm RR}[a,c]=1 ]                                                \\
		=         & \prod_{c\in B(a)} \p[ U_{\lambda,P}^\text{\rm RR}[a,c]=1 ] \cdot \prod_{c\in A\backslash B(a),c\neq a} \p[ U_{\lambda,P}^\text{\rm RR}[a,c]=1 ] \\
		=         & \left( \frac{\erm^\lambda}{1+\erm^\lambda} \right)^{|B(a)|}\cdot \left( \frac{1}{1+\erm^\lambda} \right)^{m-|B(a)|-1}                           \\
		\geqslant & \left( \frac{\erm^\lambda}{1+\erm^\lambda} \right)^{|B(b)|}\cdot \left( \frac{1}{1+\erm^\lambda} \right)^{m-|B(b)|-1}                           \\
		=         & \p[ \CMRR_\lambda(P)=b ],
	\end{align*}

	For $\CMEXP_\lambda$, we have
	\begin{align*}
		\p[ \CMEXP_\lambda(P)=a ] & =\frac{1}{1+\erm^{\lambda\cdot w_P[a,b]/2}} \cdot \prod_{c\in A\backslash\{a,b\}}\frac{1}{1+\erm^{\lambda\cdot w_P[a,c]/2}}                                      \\
		                          & \geqslant \frac{1}{1+\erm^{\lambda\cdot w_P[b,a]/2}} \cdot \prod_{c\in A\backslash\{a,b\}}\frac{1}{1+\erm^{\lambda\cdot w_P[b,c]/2}} =\p[ \CMLAP_\lambda(P)=b ],
	\end{align*}
	which means that $\CMEXP_\lambda$ satisfies p-Pareto.

	For $\CMLAP_\lambda$, we have
	\begin{align*}
		\p[ \CMLAP_\lambda(P)=a ] & =\CDF(w_P[a,b]) \cdot \prod_{c\in A\backslash\{a,b\}}\CDF(w_P[a,c])                                      \\
		                          & \geqslant  \CDF(w_P[b,a]) \cdot \prod_{c\in A\backslash\{a,b\}}\CDF(w_P[b,c])=\p[ \CMLAP_\lambda(P)=b ],
	\end{align*}
	which completes the proof.
\end{proof}

Unlike Condorcet and Pareto, the definition of monotonicity, strategyproofness, and participation are related to two distinct profiles. For monotonicity, we use the notion of absolute monotonicity in \cite{DBLP:conf/ijcai/Brandl0S18}.
Intuitively, in Mechanism \ref{algo: mech}, for any $a\in A$, whenever a voter $i$ switches $\succ_i$ to $\succ'_i$ by lifting $a$ simply, $a$ will be more likely to defeat any $b\in A$ in the one-on-one comparisons. As a consequence, $a$ will be more likely to be the winning alternative in our $\CMRand_\lambda$. Formally, the following theorem holds.

\begin{theorem}
	\label{thm: sat-pmono}
	For any $\lambda>0$, $\CMRR_\lambda$, $\CMEXP_\lambda$, and $\CMLAP_\lambda$ satisfy a-monotonicity.
\end{theorem}

\begin{proof}
	Let $P$ and $P$ be two profiles in $\LA^n$, such that $P_{-j}=P'_{-j}$ and $\succ'_j$ is a pushup of $a\in A$ in $\succ_j$. Then we have, $S_{P'}[a,b] \geqslant S_P[a,b]$ for any $b\in A$. Hence, for $\CMRR_\lambda$,
	\begin{align*}
		\p[ \CMRR_\lambda(P')=a ] \geqslant \p[ \CMRR_\lambda(P)=a ],
	\end{align*}
	which indicates that $\CMRR_\lambda$ satisfies a-monotonicity.

	For $\CMEXP_\lambda$, we have
	\begin{align*}
		  & \p[ \CMEXP_\lambda(P')=a ] = \prod_{c\in A\backslash\{a\}} \p[ U_{\lambda,P'}^\text{\rm EXP}[a,c] =1 ]                                                                                      \\
		= & \prod_{c\in A\backslash\{a\}} \frac{1}{1+\erm^{\lambda\cdot w_{P'}[a,c]/2}} \geqslant \prod_{c\in A\backslash\{a\}} \frac{1}{1+\erm^{\lambda\cdot w_P[a,c]/2}} = \p[ \CMEXP_\lambda(P)=a ].
	\end{align*}

	For $\CMLAP_\lambda$, we have
	\begin{align*}
		\p[ \CMLAP_\lambda(P')=a ] & = \prod_{c\in A\backslash\{a\}} \CDF(w_{P'}[a,c]) \geqslant  \prod_{c\in A\backslash\{a\}} \CDF(w_P[a,c]) =\p[ \CMLAP_\lambda(P)=a ],
	\end{align*}
	which completes the proof.
\end{proof}

For strategyproofness, we use the notion of SD-strategyproofness \cite{aziz2013tradeoff}, which implies the absolute monotonicity. However, the results for our rules are not so positive.

\begin{proposition}
	\label{prop: not-sd-sp}
	$\CMRR_1$, $\CMEXP_1$, and $\CMLAP_1$ do not satisfy SD-strategyproofness.
\end{proposition}

\begin{proof}
	For $\CMRR_\lambda$, let $N=\{1,2,3\}$, $A=\{a_1,a_2,a_3,a_4,a_5\}$. Consider the profiles $P=\{\succ_1,\succ_2,\succ_3\}$ and $P'=\{\succ_1,\succ_2,\succ'_3\}$, where
	\begin{align*}
		\succ_1:~~a_3 \succ a_4 \succ a_5 \succ a_1 \succ a_2, \qquad
		\succ_2:~~a_2 \succ a_3 \succ a_4 \succ a_5 \succ a_1,        \\
		\succ_3:~~a_1 \succ a_2 \succ a_3 \succ a_4 \succ a_5, \qquad
		\succ'_3:~~a_2 \succ a_3 \succ a_4 \succ a_5 \succ a_1,
	\end{align*}
	Then for $P$, we have
	\begin{itemize}
		\item $w_P[a_1,a_2]=1$, and $w_P[a_1,a_i]=-1$ for $i=3,4,5$;
		\item $w_P[a_2,a_1]=-1$, and $w_P[a_2,a_i]=1$ for $i=3,4,5$.
	\end{itemize}
	Similarly, for $P'$, we have
	\begin{itemize}
		\item $w_{P'}[a_1,a_i]=-1$, for any $i\neq 1$;
		\item $w_{P'}[a_2,a_i]=1$, for any $i\neq 2$.
	\end{itemize}

	According to Theorem 1, we have
	\begin{align*}
		\sum\limits_{a_j\succ_3 a_3}\p[\CMRR_\lambda(P)=a_j]  & = \sum\limits_{j=1}^2 \p[\CMRR_\lambda(P)=a_j]   = \frac{\erm}{(1+\erm)^4} + \frac{\erm^3}{(1+\erm)^4}, \\
		\sum\limits_{a_j\succ_3 a_3}\p[\CMRR_\lambda(P')=a_j] & = \sum\limits_{j=1}^2 \p[\CMRR_\lambda(P')=a_j]  = \frac{1}{(1+\erm)^4} + \frac{\erm^4}{(1+\erm)^4},
	\end{align*}
	which indicates that
	\[\sum\limits_{j=1}^2 \p[\CMRR_\lambda(P)=a_j] < \sum\limits_{j=1}^2 \p[\CMRR_\lambda(P')=a_j],\]
	In other word, $\CMRR_\lambda$ is not SD-strategyproof.

	For $\CMEXP_\lambda$, let $N=\{1,2,\ldots,8\}$, $A=\{a_1,a_2,a_3,a_4,a_5\}$. Consider the following profiles:
	\begin{itemize}
		\item $P$: $a_1\succ_1 a_2\succ_1 a_3\succ_1 a_4\succ_1 a_5$, and $a_2\succ_j a_3\succ_j a_4\succ_j a_5\succ_j a_1$ for $j=2,3,\ldots,8$;
		\item $P'$: $a_2\succ_1 a_1\succ_1 a_3\succ_1 a_4\succ_1 a_5$, and $a_2\succ_j a_3\succ_j a_4\succ_j a_5\succ_j a_1$ for $j=2,3,\ldots,8$.
	\end{itemize}
	Then for $P$, we have
	\begin{itemize}
		\item $w_P[a_1,a_i]=-6$, for $i=2,3,4,5$;
		\item $w_P[a_2,a_1]=6$, and $w_P[a_2,a_i]=8$ for $i=3,4,5$.
	\end{itemize}
	Similarly, for $P'$, we have
	\begin{itemize}
		\item $w_{P'}[a_1,a_2]=-8$, and $w_P[a_1,a_i]=-6$ for $i=3,4,5$;
		\item $w_{P'}[a_2,a_i]=8$, for $i=1,3,4,5$.
	\end{itemize}
	According to Theorem 1, we have
	\begin{align*}
		\sum\limits_{a_j\succ_1 a_3}\p[\CMEXP_\lambda(P)=a_j]  & = \sum\limits_{j=1}^2 \p[\CMEXP_\lambda(P)=a_j] \\ &= \left(\frac{1}{1+\erm^{3}}\right)^4 + \frac{1}{1+\erm^{-3}}\cdot \left(\frac{1}{1+\erm^{-4}}\right)^3 \approx 0.9021, \\
		\sum\limits_{a_j\succ_1 a_3}\p[\CMEXP_\lambda(P')=a_j] & = \sum\limits_{j=1}^2 \p[\CMEXP_\lambda(P')=a_j] \\ &= \frac{1}{1+\erm^{4}}\cdot \left(\frac{1}{1+\erm^{3}}\right)^3 + \left(\frac{1}{1+\erm^{-4}}\right)^4 \approx 0.9300,
	\end{align*}
	which indicates that
	\begin{align*}
		\sum\limits_{a\succ_1 a_3}\p[\CMEXP_\lambda(P)=a] < \sum\limits_{a\succ_1 a_3}\p[\CMEXP_\lambda(P')=a],
	\end{align*}
	In other word, $\CMEXP_\lambda$ does not satisfy SD-strategyproofness.

	For $\CMLAP_\lambda$, considering the same $P$ and $P'$ as $\CMEXP_\lambda$, we have
	\begin{align*}
		\sum\limits_{a\succ_1 a_3}\p[\CMLAP_\lambda(P)=a]  & = \sum\limits_{j=1}^2 \p[\CMLAP_\lambda(P)=a] = F^4_\lambda(-6) + F_\lambda(6) \cdot F_\lambda^3(8) \approx 0.9983,  \\
		\sum\limits_{a\succ_1 a_3}\p[\CMLAP_\lambda(P')=a] & = \sum\limits_{j=1}^2 \p[\CMLAP_\lambda(P')=a] = F_\lambda(-8)\cdot F_\lambda^3(-6) + F_\lambda^4(8) \approx 0.9993,
	\end{align*}
	which indicates that
	$$\sum\limits_{a\succ_1 a_3}\p[\CMLAP_\lambda(P)=a] < \sum\limits_{a\succ_1 a_3}\p[\CMLAP_\lambda(P')=a],$$
	In other word, $\CMLAP_\lambda$ is not SD-strategyproof. That completes the proof.
\end{proof}

Similar to Definition \ref{def: alpha-pCond}, we extend the notion of SD-strategyproofness.

\begin{definition}[\bf\boldmath $\alpha$-SD-Strategyproofness]
	\label{def: alpha-SDSP}
	A voting rule $r$ satisfies $\alpha$-SD-strategyproofness ($\alpha$-SD-SP for short) if for all $P,P'$ and $j\in N$, such that $P_{-j}=P'_{-j}$ and $\succ_j\neq \succ'_j$,
	$$\sum\limits_{b\succ_j a} \p[r(P)=b]\geqslant \alpha\cdot \sum\limits_{b\succ'_{j}a} \p[r(P')=b],~~ \text{for all } a\in A.$$
\end{definition}
Especially, $\alpha$-SD-strategyproofness reduces to the standard SD-strategyproofness when $\alpha = 1$. For our rules, the following theorem holds.

\begin{theorem}
	For any $\lambda>0$, $\CMRR_\lambda$, $\CMLAP_\lambda$ and $\CMEXP_\lambda$ satisfy $\erm^{(2-2m)\lambda}$-SD-strategyproofness.
	\label{thm: alpha-sd-sp}
\end{theorem}

\begin{proof}
	W.l.o.g., for any neighboring profiles $P=\{\succ_1,\succ_2,\ldots,\succ_n\}, P'=\{\succ_1',\succ_2,\ldots,\succ_n\}$ and $a\in A$, we have
	\begin{align*}
		\frac{\p[\CMLAP_\lambda(P)=a]}{\p[\CMLAP_\lambda(P')=a]}\geqslant \left( \min_{P,P'}\frac{\CDF(w_P[a_1,a])}{\CDF(w_{P'}[a_1,a])} \right)^{2(m-1)}\geqslant \erm^{(2-2m)\lambda}.
	\end{align*}
	Therefore, for any $a\in A$,
	\begin{align*}
		\sum\limits_{b\succ_1 a}\p[\CMLAP_\lambda(P)=b] \geqslant \erm^{(2-2m)\lambda} \cdot \sum\limits_{b\succ_1' a}\p[\CMLAP_\lambda(P')=b],
	\end{align*}
	which indicates that $\CMLAP_\lambda$ satisfies $\erm^{(2-2m)\lambda}$-SD-strategyproofness.

	For $\CMEXP_\lambda$, we have
	\begin{align*}
		          & \frac{\p[\CMEXP_\lambda(P)=a]}{\p[\CMEXP_\lambda(P')=a]} \geqslant \left(\min_{P,P'} \prod_{b\in A\backslash\{a\}} \frac{\erm^{\lambda\cdot w_P[a,b]/2}}{1+\erm^{\lambda\cdot w_P[a,b]/2}}\cdot \frac{1+\erm^{\lambda\cdot w_{P'}[a,b]/2}}{\erm^{\lambda\cdot w_{P'}[a,b]/2}} \right)^2 \\
		\geqslant & \left( \lim\limits_{x\to -\infty} \frac{\erm^{\lambda}(1+\erm^{\lambda\cdot x/2})}{1+\erm^{\lambda (2+x)/2}} \right)^{2-2m} = \erm^{(2m-2)\lambda}.
	\end{align*}
	Then for any $a\in A$,
	\begin{align*}
		\sum\limits_{b\succ_1 a}\p[\CMEXP_\lambda(P)=b] \geqslant \erm^{(2-2m)\lambda} \cdot \sum\limits_{b\succ_1' a}\p[\CMEXP_\lambda(P')=b].
	\end{align*}
	In other words, $\CMEXP_\lambda$ satisfies $\erm^{(2-2m)\lambda}$-SD-strategyproofness.

	Finally, we prove for $\CMRR_\lambda$. For any $a\in A$, we have
	\begin{align*}
		\frac{\p[\CMRR_\lambda(P)=a]}{\p[\CMRR_\lambda(P')=a]} \geqslant \left( \frac{\min\limits_{P\in\LA^n}\erm^{\lambda |\{b\in A: U_P[a,b]=1\}|}}{\max\limits_{P'\in\LA^n}\erm^{\lambda |\{b\in A: U_{P'}[a,b]=1\}|}} \right)^2 =\erm^{(2-2m)\lambda}.
	\end{align*}
	Then for any $a\in A$,
	\begin{align*}
		\sum\limits_{b\succ_1 a}\p[\CMRR_\lambda(P)=b] \geqslant \erm^{(2-2m)\lambda} \cdot \sum\limits_{b\succ_1' a}\p[\CMRR_\lambda(P')=b],
	\end{align*}
	which completes the proof.
\end{proof}

Finally, we discuss the participation of our voting rules. We use the notion of lexi-participation, which requires that a participating agent is always no worse off under lexicographical order. In our rules, each participating voter $j$ can benefit herself, since the majority margin $w[a,b]$ for any $a\succ_j b$ will increase due to her vote. Formally, the following theorem holds.

\begin{theorem}
	\label{thm: sat-participation}
	For any $\lambda>0$, $\CMLAP_\lambda$, $\CMEXP_\lambda$,  and $\CMRR_\lambda$ satisfy lexi-participation.
\end{theorem}

\begin{proof}
	Let $P$ and $P'$ be two profiles satisfying $P'=P\backslash \{\succ_j\}$. Suppose there is a nonempty set
	\begin{align*}
		B=\{a\in A: \p[\CMRand_\lambda(P)=a]<\p[\CMRand_\lambda(P')=a] \}.
	\end{align*}

	Let $a_{top}$ denote the top-ranked alternative in $\succ_j$. According to the definition, we have
	\begin{align*}
		S_P[a_{top}, a]=S_{P'}[a_{top}, a]+ 1,\text{ for any } a\in A\backslash\{a_{top}\}.
	\end{align*}
	Then for any $\texttt{Rand}\in \RandSet$, we have
	\begin{align*}
		\p[\CMRand_\lambda(P)=a_{top}]\geqslant \p[\CMRand_\lambda(P')=a_{top}],
	\end{align*}
	which indicates that $a_{top}\notin B$. Now, select the top-ranked alternative in $B$, denoted by $b_{top}$. Then $a_{top}\succ_j b_{top}$. Since $b_{top}\in B$, we have
	\begin{align*}
		\prod_{b\in A\backslash\{b_{top}\}} \p[U^\texttt{Rand}_{\lambda, P}[b_{top},b]=1] & = \p[\CMRand_\lambda(P)=b_{top}]                                                                                      \\
		                                                                                  & <\p[\CMRand_\lambda(P')=b_{top}] = \prod_{b\in A\backslash\{b_{top}\}} \p[U^\texttt{Rand}_{\lambda, P}[b_{top},b]=1].
	\end{align*}
	Then there must exist some $b\succ_j b_{top}$, such that
	\begin{align*}
		\p[U^\texttt{Rand}_{\lambda, P}[b_{top},b]=1] < \p[U^\texttt{Rand}_{\lambda, P'}[b_{top},b]=1].
	\end{align*}
	Noticing that $b_{top}$ is the top-ranked alternative in $B$, we have $b\notin B$, i.e.,
	\begin{align*}
		\p[\CMRand_\lambda(P)=b]\geqslant \p[\CMRand_\lambda(P')=b].
	\end{align*}
	If $\p[\CMRand_\lambda(P)=b]= \p[\CMRand_\lambda(P')=b]$, there must exist another alternative $b'\succ_j b$, such that $\p[\CMRand_\lambda(P)=b']\geqslant \p[\CMRand_\lambda(P')=b']$. Next, if $\p[\CMRand_\lambda(P)=b']= \p[\CMRand_\lambda(P')=b']$, there will be some $b''\succ_j b'$... This process will terminate after several rounds, as there are only finite number of alternatives in $A$. That is, there must exists an alternative $b\succ_j b_{top}$, such that
	\begin{align*}
		\p[\CMRand_\lambda(P)=b]> \p[\CMRand_\lambda(P')=b].
	\end{align*}
	In other words, for any $a\in A$ that $\p[\CMRand_\lambda(P)=a]< \p[\CMRand_\lambda(P')=a]$, there exists some $b\succ_j a$, such that $\p[\CMRand_\lambda(P)=b]> \p[\CMRand_\lambda(P')=b]$. Therefore, $\CMRand_\lambda$ satisfies lexi-participation, for all $\texttt{Rand}\in \RandSet$ and $\lambda>0$.
\end{proof}

Theorem \ref{thm: sat-participation} shows that for any $\texttt{Rand}\in\RandSet$, $\CMRand_\lambda$ will not harm any participating voter under lexicographical order. Further more, $\CMLAP_\lambda$ and $\CMEXP_\lambda$ satisfy a stronger notion, which is defined as follows.

\begin{definition}[Strong lexi-participation]
	\label{def: strong-lexi}
	A voting rule satisfies {\em strong lexi-participation} if for all $P,P'$ that $P'=P\backslash \{\succ_j\}$, there exists $a\in A$, such that $\p[r(P)=a]> \p[r(P')=a]$ and $\p[r(P)=b]= \p[r(P')=b]$ for all $b\succ_j a$.
\end{definition}

Intuitively, strong lexi-participation ensures that each voter can benefit from her vote, while lexi-participation only ensures that each voter will not be harmed by her vote. For $\CMLAP_\lambda$ and $\CMEXP_\lambda$, the following theorem holds.

\begin{theorem}
	\label{thm: EXP-LAP-lexi-participation}
	For any $\lambda>0$, $\CMLAP_\lambda$ and $\CMEXP_\lambda$ satisfy strong lexi-participation.
\end{theorem}

\begin{proof}
	Let $P, P'$ be two profiles in $\LA^n$, such that $P'=P\backslash \{\succ_j\}$. Suppose the top-ranked alternative of $\succ_j$ is $a$. Then for $\CMLAP$, we have
	\begin{align*}
		\p[\CMLAP_\lambda(P)=a] = \prod\limits_{b\in A\backslash\{a\}} \CDF(w_P[a,b]) > \prod\limits_{b\in A\backslash\{a\}} \CDF(w_{P'}[a,b]) =\p[\CMLAP_\lambda(P')=a],
	\end{align*}
	which indicates that $\CMLAP$ satisfies strong lexi-participation.

	Similarly, for $\CMEXP_\lambda$, we have
	\begin{align*}
		\p[\CMEXP_\lambda(P)=a] & = \prod\limits_{b\in A\backslash\{a\}} \frac{1}{1+\erm^{-\lambda\cdot w_P[a,b]/2}}                               \\
		                        & > \prod\limits_{b\in A\backslash\{a\}} \frac{1}{1+\erm^{-\lambda\cdot w_{P'}[a,b]/2}} =\p[\CMEXP_\lambda(P')=a],
	\end{align*}
	i.e., $\CMEXP_\lambda$ satisfies strong lexi-participation.
\end{proof}

However, $\CMRR_\lambda$ does not satisfy strong lexi-participation, since only one vote may be not able to increase the winning probability of any alternative. For example, consider the profile $P$, where the votes of all $n~(n\geqslant 3)$ voters are exactly the same, $a_1\succ a_2\succ\cdots\succ a_m$. Then for any $i_1<i_2$, we have $S_P[a_{i_1}, a_{i_2}]=n$ and $S_P[a_{i_2},a_{i_1}]=0$. For any $P'$ that $P'=P\backslash \{\succ_j\}$, we have $S_{P'}[a_{i_1}, a_{i_2}]=n-1$ and $S_{P'}[a_{i_2},a_{i_1}]=0$ for any  $i_1<i_2$. Then it follows that $U^\text{\rm RR}_{\lambda, P}[a_{i_1},a_{i_2}]=U^\text{\rm RR}_{\lambda, P'}[a_{i_1},a_{i_2}]$, for all $i_1,i_2\in \{1,2,\ldots,m\}$. As a result, we have
\begin{align*}
	\p[\CMRR_\lambda(P)=a_i]=\p[\CMRR_\lambda(P')=a_i], \text{ for all } i,
\end{align*}
which indicates that $\CMRR_\lambda$ does not satisfy strong lexi-participation.

\section{Differential Privacy as a Voting Axiom}

In Section \ref{sec: axioms}, we explore the tradeoffs between privacy and some voting axioms. In this section, differential privacy is regarded as an axiomatic property of voting rules. The relations between DP and some of the voting axioms are discussed. Our results are summarized in Figure \ref{fig: relation-}.

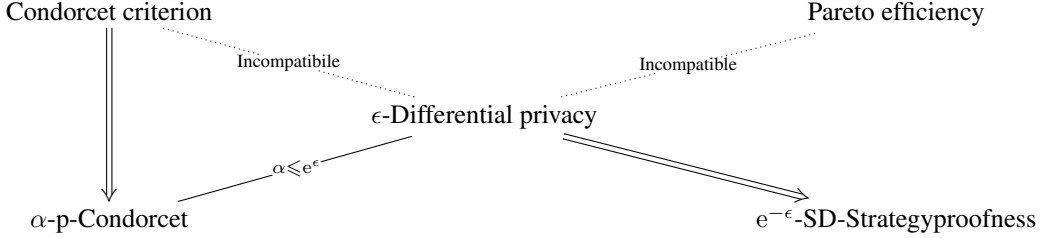
\begin{figure}[H]
	\hspace{3.5em}
	\xymatrix{
	\text{Condorcet criterion} \ar@{=>}[dd] & & & & \text{Pareto efficiency} \ar@{.}[lld]|-{\text{Incompatible}} \\
	& & \epsilon\text{-Differential privacy} \ar@{.}[llu]|-{\text{Incompatibile}} \ar@{=>}[drr] & & \\
	\alpha\text{-p-Condorcet} \ar@{-}[rru]|{\alpha\leqslant \erm^\epsilon} & & & & \erm^{-\epsilon}\text{-SD-Strategyproofness}
	}
	\caption{Relations between $\epsilon$-DP and other axioms, where $X\Rightarrow Y$ indicates that $X$ implies $Y$, a solid line between $X$ and $Y$ indicates that $X,Y$ are compatible with some condition, and a dash line between $X$ and $Y$ means that $X,Y$ are incompatible.}
	\label{fig: relation-}
\end{figure}

As proved previously, for any $\texttt{Rand}\in \RandSet$, $\CMRand_\lambda$ does not satisfy Condorcet criterion under DP. Furthermore, we can prove that they are incompatible.

\begin{proposition}
	\label{prop: Condorcet-incomp-DP}
	There is no voting rule $r$ satisfying Condorcet criterion and $\epsilon$-DP for any $\epsilon>0$.
\end{proposition}
\begin{proof}
	Suppose $r\colon \LA^n\to \Delta(A)$ is $\epsilon$-DP and satisfy Condorcet criterion. For any profile $P$ and alternative $a\in A$, we define the following operation with alternative $b\in A\backslash \{a\}$:

	\quad{\bf Op}: Choose a voter $i\in N$, whose first choice is not $b$, let $b$ be its first choice.

	\noindent Then we will obtain a profile $P'$ after at most $\lceil n/2 \rceil$ times of operations, such that $\CW(P')=b$. Suppose there has been totally $x$ times of operations, each time of operation forms a new profile, named $P_1,P_2,\ldots,P_t$ i.e.,
	$$P\xrightarrow{\text{\bf Op}}P_1\xrightarrow{\text{\bf Op}}P_2\xrightarrow{\text{\bf Op}}\cdots\xrightarrow{\text{\bf Op}}P_t=P'.$$
	Since $r$ is $\epsilon$-DP, we have
	$$\p[r(P)=a]\leqslant \erm^\epsilon \cdot \p[r(P_1)=a]\leqslant \erm^{2\epsilon}\cdot \p[r(P_2)=a]\leqslant\ldots\leqslant \erm^{t\cdot \epsilon}\cdot\p[r(P')=a].$$
	However, the Condorcet criterion indicates that $\p[r(P')=b]=1$, which follows $\p[r(P')=a]=0$. Thus, $\p[r(P)=a]=0$, for any profile $P$ and alternative $a$, which contradicts to the definition of randomized voting rule.
\end{proof}

Similarly, Pareto efficiency is also incompatible with DP, which indicates that the stronger notions of efficiency, e.g., PC-efficiency and SD-efficiency~\cite{Brandt2017:Rolling} are all incompatible with DP. Formally, we have the following result.

\begin{proposition}
	\label{prop: Pareto-incomp-DP}
	There is no voting rule $r$ satisfying Pareto efficiency and $\epsilon$-DP for any $\epsilon>0$.
\end{proposition}

\begin{proof}
	Suppose $r\colon \LA^n\to \Delta(A)$ is $\epsilon$-DP and satisfy Pareto efficiency. For any profile $P$ and alternative $a\in A$, we define the following operation:

	\quad {\bf Op}: Choose a voter $i\in N$, whose last choice is not $a$, let $a$ be its last choice.

	\noindent Then we will obtain a profile $P'$ after at most $n$ times of operations, where each voter's last choice is $a$. Suppose there has been totally $x$ times of operations, each time of operation forms a new profile, named $P_1,P_2,\ldots,P_t$ i.e.,
	$$P\xrightarrow{\text{\bf Op}}P_1\xrightarrow{\text{\bf Op}}P_2\xrightarrow{\text{\bf Op}}\cdots\xrightarrow{\text{\bf Op}}P_t=P'.$$
	Since $r$ is $\epsilon$-DP, we have
	$$\p[r(P)=a]\leqslant \erm^\epsilon \cdot \p[r(P_1)=a]\leqslant \erm^{2\epsilon}\cdot \p[r(P_2)=a]\leqslant\ldots\leqslant \erm^{t\cdot\epsilon}\cdot\p[r(P')=a].$$
	However, the Pareto efficiency for randomized voting indicates that $\p[r(P')=a]=0$, as $a$ is definitely Pareto dominated in $P'$. Thus, $\p[r(P)=a]=0$, for any profile $P$ and alternative $a$, which contradicts to the definition of randomized voting rule.
\end{proof}

In fact, the proofs for Propositions \ref{prop: Condorcet-incomp-DP} and \ref{prop: Pareto-incomp-DP} indicate that for any DP voting rule $r$, there is no profile $P$ and alternative $a$, such that $\p[r(P)=a]=0$. To measure the incompatibility between Condorcet criterion and DP, we use the notion of $\alpha$-p-Condorcet. The result is shown as follows.

\begin{proposition}
	\label{prop: DP-Cond-upperbound}
	There is no voting rule satisfying $\epsilon$-DP and $\alpha$-p-Condorcet with $\alpha>\erm^{\epsilon}$.
\end{proposition}

\begin{proof}
	Let $P,P'$ be profiles that $\CW(P)=a$,$\CW(P')=b$, $P_{-j}=P'_{-j}$, and $\succ_j\neq \succ_j'$. Then
	\begin{align*}
		          & \p[f(P)=a] \geqslant \alpha\cdot \p[f(P)=b] \geqslant \alpha\cdot \erm^{\epsilon}\cdot \p[f(P')=b]        \\
		\geqslant & \alpha^2\cdot\erm^{\epsilon}\cdot \p[f(P')=a] \geqslant \alpha^2\cdot \erm^{-2k\epsilon}\cdot \p[f(P)=a].
	\end{align*}
	Thus, $\alpha^2\erm^{-2\epsilon}\leqslant 1$, i.e., $\alpha\leqslant \erm^{\epsilon}$.
\end{proof}

The SD-strategyproofness is compatible with DP, as the trivial voting rule, i.e., $\p[r(P)=a]=1/m$, for all $a\in A$, satisfies SD-strategyproofness and $0$-DP. In fact, DP admits a lower bound of satisfaction to strategyproofness. To be more precise, we use the notion of $\alpha$-SD-strategyproofness.

Then the following proposition holds.

\begin{proposition}
	\label{prop: dp-strategyproof}
	Any voting rule satisfying $\epsilon$-DP satisfies $\erm^{-\epsilon}$-SD-strategyproofness.
\end{proposition}

\begin{proof}
	Suppose $r$ is a voting rule satisfying $\epsilon$-DP,  $P$ and $P'$ are profiles differing on only one voter's ballot. Then, for any $j\in N$ and any $a\in A$, DP indicates that
	\begin{align*}
		\p[r(P)=a]\leqslant \erm^{-\epsilon}\cdot \p[r(P')=a].
	\end{align*}
	Then we have
	\begin{align*}
		\sum\limits_{b\succ_i a} \p[r(P)=b] & \leqslant  \erm^{-\epsilon}\cdot \sum\limits_{b\succ_i a} \p[r(P')=b],
	\end{align*}
	which completes the proof.
\end{proof}

As is shown in Proposition \ref{prop: dp-strategyproof}, the satisfaction to strategyproofness increases when $\epsilon$ decreases. This is quite intuitive, since there is little motivation for an adversary to manipulate the voting process when the outcomes of neighboring datasets are very similar.

\section{Conclusion and Future Work}
\label{sec: conclusion}
In the paper, we proposed three classes of differentially private Condorcet methods and explored their accuracy-privacy tradeoff and axioms-privacy tradeoff. Further, we investivated the relations between DP and other axioms. For future work, we plan to explore more axioms for randomized voting rules and design new voting rules performing better on satisfaction to the axioms. The design and analysis of differentially private mechanisms for other social choice problems, such as multi-winner elections, fair division, and participatory budgeting, are also promising directions for future work.

\bibliography{aaai23}

\end{document}